\documentclass{article}

\usepackage[preprint]{neurips_2022}

\usepackage[utf8]{inputenc} % allow utf-8 input
\usepackage[T1]{fontenc}    % use 8-bit T1 fonts
\usepackage[hidelinks]{hyperref}       % hyperlinks
\usepackage{url}            % simple URL typesetting
\usepackage{booktabs}       % professional-quality tables
\usepackage{amsfonts}       % blackboard math symbols
\usepackage{nicefrac}       % compact symbols for 1/2, etc.
\usepackage{microtype}      % microtypography
\usepackage{xcolor}         % colors
\usepackage[disable]{todonotes}
\usepackage{xspace}

\usepackage{amsmath}
\usepackage{amsthm}
\usepackage{amssymb}
\usepackage{enumitem}
\theoremstyle{plain}
\newtheorem{theorem}{Theorem}[section]

\newtheorem{lemma}[theorem]{Lemma}

\theoremstyle{definition}
\newtheorem{definition}[theorem]{Definition}

\theoremstyle{remark}

\newtheorem*{example*}{Example}
\newtheorem{remark}{Remark}
\newtheorem*{remark*}{Remark}

\usepackage{xparse}
\usepackage{ifthen}

\newcommand{\braced}[1]{\ifthenelse{\equal{#1}{}}{}{\left(#1 \right)}}
\newcommand{\sbraced}[1]{\ifthenelse{\equal{#1}{}}{}{(#1)}}

\newcommand{\diff}[1]{\frac{d #1}{dt}}
\NewDocumentCommand{\R}{}{\mathbb{R}}

\NewDocumentCommand{\ti}{}{t}

\NewDocumentCommand{\states}{}{S}
\NewDocumentCommand{\st}{}{s}
\NewDocumentCommand{\stb}{}{s'}
\NewDocumentCommand{\players}{}{I}
\NewDocumentCommand{\pl}{}{i}
\NewDocumentCommand{\plb}{}{j}
\NewDocumentCommand{\trans}{O{s} O{} O{}}{P_{#1\ifthenelse{\equal{#2}{}}{}{ #2}}\ifthenelse{\equal{#3}{}}{}{( #3 )}}
\NewDocumentCommand{\actions}{O{}}{A^{#1}}
\NewDocumentCommand{\poff}{O{} O{\pl} O{\st}}{r^{#2}_{#3}\ifthenelse{\equal{#1}{}}{}{\left( #1 \right)}}
\NewDocumentCommand{\simplex}{m}{\Delta\left({#1}\right)}

\NewDocumentCommand{\tble}{O{\auxact} O{\pl}}{h^{#2}\braced{#1}}
\NewDocumentCommand{\smness}{}{\epsilon}

\NewDocumentCommand{\auxpoff}{O{\ti} O{\pl} O{\auxact[#1]} O{\st} O{\auxval[#1][#2][]}}{f^{#2}_{#4, #5}\braced{#3}}
\NewDocumentCommand{\auxact}{O{\ti} O{} O{\st}}{x^{#2}_{#3}\sbraced{#1}}
\NewDocumentCommand{\dotauxact}{O{} O{} O{\st}}{\dot{x}^{#2}_{#3} \ifthenelse{\equal{#1}{}}{}{\left( #1 \right)}}
\NewDocumentCommand{\auxval}{O{\ti} O{\pl} O{\st}}{u^{#2}_{#3}\sbraced{#1} }
\NewDocumentCommand{\dotauxval}{O{} O{\pl} O{\st}}{\dot{u}^{#2}_{#3} \ifthenelse{\equal{#1}{}}{}{\left( #1 \right)}}
\NewDocumentCommand{\rateval}{O{\ti} O{\st}}{\beta\sbraced{#1}}
\NewDocumentCommand{\rateact}{O{\ti} O{\st}}{\alpha_{#2}(#1)}
\NewDocumentCommand{\sbr}{O{\st} O{\auxval[\ti][\pl][]} O{\auxact} O{\pl}}{\operatorname{sbr}^{#4}_{#1, #2}\braced{#3}}
\NewDocumentCommand{\argmax}{m}{\arg\max_{#1}}

\NewDocumentCommand{\esttrans}{O{s} O{} O{} O{\ti}}{\hat P_{#1 #2}\braced{#3}\braced{#4}}
\NewDocumentCommand{\estpoff}{O{\act} O{\pl} O{\st} O{\ti}}{\hat r^{#2}_{#3}\braced{#1}\braced{#4}}

\NewDocumentCommand{\dotesttrans}{O{s} O{} O{} O{}}{\dot {\hat {P}}_{#1 #2}\braced{#3}\braced{#4}}
\NewDocumentCommand{\dotestpoff}{O{\act} O{\pl} O{\st} O{}}{\dot {\hat {r}}^{#2}_{#3}\braced{#1}\braced{#4}}

\NewDocumentCommand{\estauxpoff}{O{\ti} O{\pl} O{\auxact[#1]} O{\st} O{\auxval[#1][#2][]}}{{\hat f}^{#2}_{#4, #5}\braced{#3}}

\NewDocumentCommand{\nti}{}{n+1}
\NewDocumentCommand{\dauxval}{O{\ti} O{\pl} O{\st}}{u^{#2}_{\ifthenelse{\equal{#3}{}}{#1}{\ifthenelse{\equal{#1}{}}{#3}{#3,#1}}}}
\NewDocumentCommand{\ddotauxval}{O{\pl} O{\st}}{\dauxval[\nti][#1][#2] - \dauxval[\ti][#1][#2]}
\NewDocumentCommand{\speedval}{}{\beta} 
\NewDocumentCommand{\drateval}{O{\ti} O{\st}}{ \frac{\speedval}{#1+1}} %\gamma_{#2, #1}}
\NewDocumentCommand{\drateact}{O{\ti} O{\st}}{\frac{1_{#2 = \curst}}{\stcount[\ti][\st]}}%\gamma_{#2, #1}}
\NewDocumentCommand{\curst}{O{\ti}}{s_{#1}}
\NewDocumentCommand{\stcount}{O{\ti} O{\st}}{#2^\sharp_{#1}}
\NewDocumentCommand{\dsbr}{O{\st} O{\dauxval[\dti][][]} O{\dauxact} O{\pl}}{\operatorname{sbr}^{#4}_{#1, #2}\braced{#3}}

\NewDocumentCommand{\dauxact}{O{\dti} O{} O{\st}}{x^{#2}_{\ifthenelse{\equal{#3}{}}{#1}{#3, #1}}}
\NewDocumentCommand{\ddotauxact}{O{} O{\st}}{\dauxact[\nti][#1][#2]-\dauxact[\ti][#1][#2]}

\NewDocumentCommand{\act}{O{\ti} O{} O{} O{}}{a^{#2}_{#3}\braced{#1}\braced{#4}}
\NewDocumentCommand{\dact}{O{\dti} O{} O{} O{}}{a^{#2}_{\ifthenelse{\equal{#3}{}}{#1}{#3, #1}}}
\NewDocumentCommand{\prof}{O{} O{\st}}{x^{#1}_{#2}}
\NewDocumentCommand{\profb}{O{} O{\st}}{y^{#1}_{#2}}
\NewDocumentCommand{\purep}{O{} O{\st}}{a^{#1}_{#2}}
\NewDocumentCommand{\purepb}{O{} O{\st}}{b^{#1}_{#2}}
\NewDocumentCommand{\disc}{}{\delta}

\NewDocumentCommand{\dti}{}{n}

\NewDocumentCommand{\desttrans}{O{s} O{} O{} O{\dti}}{\hat P_{#1 #2, #4}\braced{#3}}
\NewDocumentCommand{\destpoff}{O{\dact} O{\pl} O{\st} O{\dti}}{\hat r^{#2}_{#3, #4}\braced{#1}}
\NewDocumentCommand{\destauxpoff}{O{\dti} O{\pl} O{\dauxact[#1]} O{\st} O{\dauxval[#1][#2][]}}{{\hat f}^{#2}_{#4, #5}\braced{#3}}

\NewDocumentCommand{\tipoff}{O{\ti} O{\pl}}{R_{#2}^{#1}}
\NewDocumentCommand{\distrpoff}{O{} O{\pl} O{\st}}{\tilde{r}^{#2}_{#3}\ifthenelse{\equal{#1}{}}{}{\left( #1 \right)}}
\NewDocumentCommand{\dtble}{O{\dauxact} O{\pl}}{h^{#2}\braced{#1}}

\NewDocumentCommand{\ratemin}{}{\alpha_-}

\NewDocumentCommand{\ratelim}{}{\beta^\star}

\title{Smooth Fictitious Play in Stochastic Games with Perturbed Payoffs and Unknown Transitions}

\author{%
  Lucas Baudin \\
  Université Paris-Dauphine - PSL\\
  \texttt{lucas.baudin@dauphine.eu} \\
  \And
  Rida Laraki \\
  Université Paris-Dauphine - PSL
}

\newif\ifexample
\examplefalse

\begin{document}

\maketitle

\begin{abstract}
  Recent extensions to dynamic games (\cite{leslieBestresponseDynamicsZerosum2020, sayinFictitiousPlayZerosum2020, baudinFictitiousPlayBestResponse2022}) of the well-known fictitious play learning procedure in static games were proved to globally converge to stationary Nash equilibria in two important classes of dynamic games (zero-sum and identical-interest discounted stochastic games). However, those decentralized algorithms need the players to know exactly the model (the transition probabilities and their payoffs at every stage). To overcome these strong assumptions, our paper introduces regularizations of the systems in \cite{leslieBestresponseDynamicsZerosum2020, baudinFictitiousPlayBestResponse2022} to construct a family of new decentralized learning algorithms which are model-free (players don't know the transitions and their payoffs are perturbed at every stage). Our procedures can be seen as extensions to stochastic games of the classical smooth fictitious play learning procedures in static games (where the players best responses are regularized, thanks to a smooth strictly concave perturbation of their payoff functions). We prove the convergence of our family of procedures to stationary regularized Nash equilibria in zero-sum and identical-interest discounted stochastic games. The proof uses the continuous smooth best-response dynamics counterparts, and stochastic approximation methods. When there is only one player, our problem is an instance of Reinforcement Learning and our procedures are proved to globally converge to the optimal stationary policy of the regularized MDP. In that sense, they can be seen as an alternative to the well known Q-learning procedure.
  
  %Following recent extensions of fictitious play to identical interests and zero-sum stochastic games, we study an extension of smooth stochastic games in the model-free setting where the transition probability between states is estimated and payoffs are perturbed. Convergence results are established using the continuous counterpart, that is smooth best-response dynamics, for two classes of stochastic games: zero-sum and identical-interest games.
\end{abstract}

\section{Introduction}

Fictitious play (FP) (\cite{brown1951iterative,robinsonIterativeMethodSolving1951}) is the oldest and most famous independent (i.e., decentralized) learning algorithm in game theory. It is a simple behavioural strategy that agents may use to repeatedly play a fixed normal form game $G$. Namely, at every repetition of $G$, each player best replies to the empirical distribution of the uncorrelated past actions of the opponents.

When $G$ is a finite game, it has been shown that the empirical distribution of the players' actions
under FP globally converges to the set of mixed Nash equilibria, when $G$ is zero-sum (\cite{brown1951iterative,robinsonIterativeMethodSolving1951}) or identical-interest (\cite{mondererFictitiousPlayProperty1996}).
%and monotone games(\cite{hofbauerGlobalConvergenceStochastic2002}). 
The FP procedure was extensively studied since then with numerous additional convergence results (\cite{RePEc:cla:levarc:419, bergerFictitiousPlayGames2005, vanstrienFictitiousPlayGames2011}) and a number of generalizations (\cite{leslieGeneralisedWeakenedFictitious2006, benaimConsistencyVanishinglySmooth2013} to cite a few).

Unfortunately, FP has no guarantee in terms of regret. Indeed, it may well lead to the worst possible payoff in a zero-sum game if the opponent best-replies to it at every stage. Fortunately, the smooth version of FP (call it SFP), also named stochastic or regularized fictitious play, is proved to have the no-regret property (\cite{fudenbergConsistencyCautiousFictitious1995,benaimConsistencyVanishinglySmooth2013}) and to converge to regularized Nash Equilibria in the classes mentioned so far (\cite{hofbauerGlobalConvergenceStochastic2002, bergerFictitiousPlayGames2005,cohenLearningBanditFeedback}, even when $G$ is non atomic or with continuous action sets, see \cite{hadikhanlooLearningNonatomicGames2021,perrinFictitiousPlayMean2020}). In the online optimization and learning community, fictitious play for two player games (or correlated FP for more than two players, where a player best responds to the correlated empirical past actions of the opponents) is precisely Follow The Leader (FTL), and Smooth Fictitious Play is Follow The Regularized Leader (FTRL) (see \citet{shalev-shwartzOnlineLearningOnline2011a, cesa-bianchiPotentialBasedAlgorithmsOnline2001,belmegaOnlineConvexOptimization2018a,kwonContinuoustimeApproachOnline2017,GVM21,GVM21b}).

%On the other hand, there are finite games where nor FP, or SFP or any independent learning rule can globally converge to the set of Nash equilibria of $G$ (see \cite{ shapleyTopicsTwoPersonGames1964,10.2307/3132156,hofbauer_sigmund_1998,DEMICHELIS2000192,GSSV2013}).

%Dual Averaging (DA) (REF) is another popular class of no-regret algorithms proved to converge to the set of mixed Nash equilibria whenever the underlying static game is zero-sum (REF), identical-interest (REF), or monotone (REF). 

%(\textbf{REF SHAPLEY Game / Triangle}) and it is well known that no decentralized learning algorithm can globally converge to the set of Nash equilibria in every game (REF Hart-Mas Collel, Hofbauer-Sigmund, De Michelis.).

How can FP and SFP be extended to learn Nash equilibria in zero-sum and identical-interest games when the stage game is not fixed but changes with time (such as the stochastic game of \citet{shapleyStochasticGames1953})?

The difficulty is that, in stochastic games (SG), a state variable evolves from a period to the next one with a probability depending on the current state and actions players take. Therefore, every player, when choosing its action at any given stage, must strike a balance between its instantaneous reward (determined by the current state and the current profile of actions) and its continuation future rewards, which depend in particular on the next state. Extending the result of \citet{shapleyStochasticGames1953} in the zero-sum case, \citet{finkEquilibriumStochasticPerson1964} proved that when there are finitely many players and actions (the framework of our paper), any discounted stochastic game (DSG) admits a stationary Nash equilibrium in mixed strategies (e.g. a decentralized randomized policy that depends only on the state variable).

Stationary Nash are the simplest possible equilibria that a dynamic game can have and it is our objective to construct independent learning algorithms that are provable to globally converge to them in finite zero-sum and identical-interest DSG. This is a hard problem because the payoff function of a player in a DSG is not own-payoff concave or quasi-concave when the players are restricted to their stationary strategies and thus, no gradient based method is guaranteed to converge, even to a local Nash equilibrium (\cite{daskalakisComplexityConstrainedMinmax2021}). 
%The design of decentralized (e.g. independent learning) algorithms that globally converge to stationary Nash equilibria have received much less attention when the stage payoffs vary stochastically with time. 
Only recently some positive results have been obtained. That is, \citet{leslieBestresponseDynamicsZerosum2020, sayinFictitiousPlayZerosum2020,baudinFictitiousPlayBestResponse2022} combined the Fictitious Play behavioral strategy with a Q-learning like updating rule to design a family of decentralized rules such that the sequence of independent empirical distributions of the players converge to the set of stationary Nash equilibria in the two major classes of zero-sum and identical-interest ergodic DSG.

These latest results are based on ideas of the well-known and comparatively much better understood framework of Reinforcement Learning (RL). In this setting, there is only one player (a Markov Decision Process). The advent of efficient and model-free algorithms in this context such as Q-learning (\cite{watkinsLearningDelayedRewards1989}) has had a lot of impacts on RL with numerous extensions, including offline Q-learning (\cite{NEURIPS2020_0d2b2061}), double Q-learning (\cite{NIPS2010_091d584f}) and a wide range of applications (\cite{taiRobotExplorationStrategy2016, kurinCanQlearningGraph2020}). However, the convergence of Q-learning does not extend to the multiagent setting (see \cite{wunderClassesMultiagentQlearning2010,kianercyDynamicsBoltzmannLearning2012}).
%Also, while by using Q-learning in the RL context, the decision maker only learns the optimal continuation payoff and an optimal policy, our algorithms additionally succeed in learning the model. 

Unfortunately, the algorithms in \citet{leslieBestresponseDynamicsZerosum2020} and \citet{baudinFictitiousPlayBestResponse2022} need the players to know precisely the model from the beginning (i.e., the transition and payoff functions). To avoid this drawback, our paper introduces a robust regularized version of their rules which combines Smooth FP, Q-learning-like rule (i.e., updates at every step an estimate of the continuation payoff) and empirical estimates of the unknown parameters. This leads us to a family of model-free independent learning algorithms, where the payoffs and the transitions are unknown to the players, and where the stage payoffs are imperfectly observed (e.g. randomly perturbed with a zero-mean noise). Our algorithms are proved to converge to the set of regularized stationary Nash equilibria in zero-sum and identical-interest ergodic DSG (those are a subset of stationary $\varepsilon$-Nash equilibria, and as $\varepsilon$ goes to zero (corresponding to a vanishing SFP), they refine the set of stationary Nash equilibria).

%\footnote{The algorithm in \citet{sayinFictitiousPlayZerosum2020} is model-free, but was proved to converge only in the class of zero-sum SG. It is not obvious how to extend the convergence result to identical-interest DSG.}

%RAJOUTER ICI DES REF / TOPO SUR LES AUTRES PROCEDURES DANS Q-LEARNING WITH APPRENNENT AUSSI LA POLICY OPTIMAL ET LA RELATION AVEC NOUS. 

Another efficient class of algorithms in RL are the Projected Gradient Methods (PGM) and their
stochastic version (\citet{Williams1992,suttonReinforcementLearningIntroduction2018}). Very recently,
PGM and Stochastic PGM (SPGM) have been studied by \citet{NEURIPS2020_3b2acfe2} in zero-sum stochastic games to
prove that when the players play independently a PGM (or a SPGM) with different time
scales,\footnote{ \citet{NEURIPS2020_3b2acfe2} proved that their convergence result fails if the players
have the same time scale.} one can approach a best stationary equilibrium iterate\footnote{I.e., there is
$t\in \{1,...,T(\varepsilon) \}$ such that (in expectation in case of Stochastic PGM), the players play
an $\epsilon$ stationary Nash equilibrium at time $t$.} in a finite time $T(\varepsilon)$, whenever the
stochastic game is episodic.\footnote{A SG is episodic if, at every stage, there is a positive
probability that the game stops.} \citet{LPOP2022} proved a similar convergence result in episodic
identical-interest DSG, as soon as all the players use a PGM or SPGM.\footnote{\citet{LPOP2022} proved
their results in the larger class of potential stochastic games (they called Markov Potential Games). In
such a class of games, the players are divided into two categories: either they do not influence the
transition or they have the same payoff function up to a constant (see
\citet{hollerLearningDynamicsReinforcement2020,LPOP2022}). This last class is slightly larger than
identical-interest SG is called \emph{Team Stochastic Games} in
\citet{hollerLearningDynamicsReinforcement2020}. All our convergence results extend to team stochastic
games. The class of potential SG where the transitions are independent of the player's actions is easy
to solve, just let each player uses a smooth FP myopically state per state independently.} Those are
extremely interesting and promising convergence results. Note however that they don't cover our
framework because they assume the game to have a positive probability to stop (implying that their
trajectories terminate almost surely in finite time) while in our ergodic setting, all the trajectories
are infinite and the game never terminates. So, our results and those of
\citet{NEURIPS2020_3b2acfe2,LPOP2022} are non-comparable and complementary: they prove the existence of
a best iterate approximation, we prove a time-average convergence, we use different algorithms and
tools, and orthogonal assumptions (episodic vs ergodic).

The counterpart of FP in continuous time is best-response dynamics (\cite{matsuiBestResponseDynamics1992,harrisRateConvergenceContinuousTime1998a}) and modern proofs for the convergence of SFP rely on the convergence of the continuous model combined with stochastic approximation techniques (see for instance \cite{benaimStochasticApproximationsDifferential2005, benaimConsistencyVanishinglySmooth2013, hadikhanlooLearningNonatomicGames2021}). We follow a similar approach: we prove the convergence to regularized stationary Nash equilibria of an associated smooth continuous-time dynamics and, using some advanced stochastic approximation tools, deduce the convergence of our discrete-time rules. This is an important technical difference with \citet{NEURIPS2020_3b2acfe2,LPOP2022} who can prove their results directly in discrete time, with an explicit finite convergence bound $T(\varepsilon)$. Our stochastic approximation methods do not provide us with a convergence rate, which is an open problem for SFP and FP even in the classical setting.

\paragraph{Contributions}
\begin{itemize}
\item We introduce a family of independent learning algorithms in stochastic games that are model-free (unknown transitions and imperfect observation of the stage payoffs);
\item We identify the corresponding smooth continuous-time dynamics, and show it globally converges to regularized stationary Nash equilibria in identical-interest and zero-sum DSG;
\item From the continuous-time convergence, we deduce that in our two classes of DSG, the uncorrelated empirical frequencies of actions generated by our discrete-time algorithms almost surely globally converge to the set of regularized stationary Nash equilibria.
\end{itemize}

\paragraph{Outline} Section~\ref{sec:preliminaries} gives the main definitions and notations of the paper. Section~\ref{sec:sfp} introduces the smooth fictitious play procedure together with the continuation payoff updating rule. In order to prove the convergence of the discrete time algorithm, and also for its own sake, a smooth best-response dynamics with a simple continuation payoff up-dating rule dynamics is described in Section~\ref{sec:sbrd} followed with sketches of proofs of both our discrete-time and continuous-time systems. Section~\ref{sec:related} describes in more details the related work. The appendix contains the detailed proofs. \ifexample Section~\ref{sec:soccer} explains an example and presents empirical results of our algorithm.\fi

\section{Preliminaries}\label{sec:preliminaries}

{
\RenewDocumentCommand{\act}{}{\dact}
\RenewDocumentCommand{\ti}{}{n}

Strategic situations where several agents interact, get rewards and modify an environment can be modeled as Stochastic Games (SG). In our settings, the number of players, actions and possible states are finite. We consider players that are interested in the so-called discounted reward on an infinite horizon, that is players strike a balance between instantaneous rewards and future ones.

\paragraph{Stochastic games} SG are tuples
  \(G=(\states, \players, \actions, \{\poff\}_{\pl \in \players, \st \in \states}, \{\trans\}_{\st \in \states})\)
  where \(S\) is the state space (a finite set), \(I\) is the finite set of
  players, \(\actions[\pl]\) is the finite action set of player \(\pl\),
  \(\actions := \Pi_{\pl\in I} \actions[\pl]\) is the set of action profiles,
  \(\poff (\cdot): \actions \rightarrow \R\) is the stage reward of player
  \(\pl\),
  and \(\trans (\cdot): \actions \rightarrow \simplex{\states}\) is the transition probability
  map (where $\simplex{\states}$ is the set of probability distributions on $\states$).
%\end{definition}

\NewDocumentCommand{\mixedactions}{}{X}

%A (pure) action profile is an element of ${\actions[]}^\states$ with a typical element $\{\purep[\pl]\}_{\pl \in \players, \st \in \states}$. Mixed actions for player $i$ at a given state are probability distributions over its actions, written $\simplex{\actions[\pl]}$. Set $\mixedactions$ denotes mixed actions for every player $\Pi_{\pl \in \players}\simplex{\actions[\pl]}$ and a mixed actions profile is an element of $\mixedactions^\states$.

\paragraph{Main Restrictions} We are interested in two classes of games: \emph{zero-sum} stochastic games are two-players SG where $\poff[][1] = - \poff[][2]$ for every $\st$ and \emph{team} stochastic games are such that the payoff functions of the players differ only by a constant (there is $r_s(\cdot):A\rightarrow \R$ such that for every $i$ and $s$, $r^i_s(\cdot)=r_s(\cdot)+c^i$ for a constant $c^i$). A special case is \emph{identical-interest} SG where all payoff functions are equal ($r^i_s(\cdot)=r_s(\cdot)$ for all $i$). A SG is \emph{ergodic} if there exists $T \in \mathbb N$ such that for any sequence of actions of length $T$, the probability to reach any state $\stb$ starting from any state $\st$ is positive.

\paragraph{Game Form} A stochastic game is played as follows: starting from an initial state $\st_0$, at every step $\ti \in \mathbb N$, every player $\pl$ choose an action $\act[\ti][\pl]$ given the history of play and the current state $\curst$. The next state $\curst[\nti]$ is drawn from distribution $\trans[\curst][][\act]$.

\paragraph{Discounted Payoffs} We suppose that every player $\pl$ is interested in maximizing its discounted payoff, that is the expectancy of $\sum_{k \in \mathbb N} \disc^k \poff[\act[k]][\pl][\curst[k]]$ where $\delta \in (0, 1)$ is the discount factor.

\paragraph{Strategies} A behavioral strategy $\sigma^i$ for player $i$ is a mapping associating with %to 
each stage $n\in \mathbb{N}$, history $h_n\in (S\times A)^n$ and current state $s$, a mixed action $x^i_n=\sigma^i(n,h_n,s)$ in $\Delta(A^i)$. The behavioral strategy is pure if its image is always in $A^i$. A stationary strategy of player \(i\) is the simplest of behavioral strategies. It depends only on the current state $s$ but not on the period $n$ nor on the past history $h_n$. As such, a stationary strategy can be identified with an element of \(\Delta(A^i)^S\) (a mixed action per state interpreted as: whenever the state is $s$, $i$ plays randomly according to $x^i_s$). The set of stationary strategy profiles is $\Pi_{i\in I} \Delta(A^i)^S$. Set $\mixedactions=\Pi_{\pl \in \players}\simplex{\actions[\pl]}$ so a stationary profile is an element of $\mixedactions^\states$. For \(y^i \in (\Delta(A^i))^S\) and \(x \in \Pi_{i\in I} \Delta(A^i)^S\), we denote by \((y^i, x^{-i})\) the stationary profile where $i$ changes its strategy from $x^i$ to $y^i$. A stationary profile \(x \in X^S\) is a Nash equilibrium if and only if no player has
a profitable behavioral deviation. \citet{finkEquilibriumStochasticPerson1964} proved the existence of stationary Nash equilibria in every finite DSG and that it is sufficient to check pure stationary deviations. 

%A stochastic game admits typically many other Nash equilibria. To see why, if there is only one state, it is a repeated game and the famous folk theorem of repeated game \cite{aumannLongTermCompetitionGameTheoretic1994, 10.2307/1911307,larakiMathematicalFoundationsGame2019} shows that when the discount factor is large enough, any feasible and individually rational payoff of the stage game is a Nash equilibrium payoff of the repeated game.

\paragraph{Regularizer} In this paper, we are interested in exploratory algorithms, which may classically be generated by the use of some steep concave regularizer. This regularizer is added to the payoff functions $\poff[]$ and can be given several interpretations (see \cite{fudenbergTheoryLearningGames1998,hofbauerGlobalConvergenceStochastic2002} for details): it models the uncertainty of the payoff caused by the "trembling hand" of players or is a way to generate strict incentives to explore all the actions. Formally player $\pl$ maximizes a perturbation of its payoff function $\poff[] + \smness \tble[]$ under the following hypotheses: 
\begin{equation}\label{eq:tble}\tag{H1}
  \begin{gathered}
  \tble[]: \mixedactions \rightarrow \R^+\text{, strictly concave in\ } x^i,\ C^1\text{ on the interior,}\\
  \lim_{x^i\rightarrow \partial \simplex{\actions[\pl]}} \| \nabla_{x^i} \tble[x]\| = + \infty \text{\ and \ }\smness > 0 
  \end{gathered}
\end{equation} % last hypothesis similar to A2 in \cite{benaimConsistencyVanishinglySmooth2013}

In this paper we study the convergence of some discrete and continuous time systems to regularized Stationary Nash equilibria, parameterized by the regularizers $(\tble[])_{\pl\in\players}$ and parameter $\smness$.

\begin{definition}\textbf{Regularized Stationary Nash Equilibria} of a DSG are the stationary Nash equilibria of the DSG with the perturbed payoff functions $\poff[] + \smness \tble[]$, $i\in I$.
\end{definition}

%Note that for $\smness = 0$, regularized equilibria are exactly the stationary Nash equilibria.

\begin{remark*}
 In identical-interest (resp. zero-sum) SG, we suppose all players take the same (resp. the opposite) regularizer $\tble[][]$ so as the regularized payoff functions remain identical. In this context, one can take for example a separable regularizer function which is equal to the sum (resp. the difference) of concave functions depending only on $x^i$.
\end{remark*}

\begin{remark*}
  If in a profile, all actions are $\speedval$-optimal with respect to discounted payoff based on functions $\poff[] + \smness \tble[]$ (\emph{i.e.,} deviations are at most $\speedval$ profitable), then this profile is called a $\beta$-regularized equilibria.
\end{remark*}

}

\section{Smooth Fictitious Play in Stochastic Games (Discrete-Time Algorithms)}\label{sec:sfp}

{

\RenewDocumentCommand{\auxact}{}{\dauxact}
\RenewDocumentCommand{\dotauxact}{}{\ddotauxact}
\RenewDocumentCommand{\auxval}{}{\dauxval}
\RenewDocumentCommand{\rateval}{}{\drateval}
\RenewDocumentCommand{\rateact}{}{\drateact}
\RenewDocumentCommand{\dotauxval}{}{\ddotauxval}
\RenewDocumentCommand{\ti}{}{n}
\RenewDocumentCommand{\estauxpoff}{}{\destauxpoff}
\RenewDocumentCommand{\estpoff}{}{\destpoff}
\RenewDocumentCommand{\esttrans}{}{\desttrans}
\RenewDocumentCommand{\act}{}{\dact}

\paragraph{Fictitious Play in Repeated Games} Introduced by \citet{robinsonIterativeMethodSolving1951} and \citet{brown1951iterative}, fictitious play is a decentralized behavioral strategy to repeatedly play a fixed normal form game. At every step $\ti$, every player $\pl$ chooses an action $\act[\nti][\pl]$ that is a best response to the past empirical average action of other players: $\act[\nti][\pl]$ must maximize $\poff[\cdot, \auxact[\ti][-\pl][]][\pl][]$ where $\auxact[\ti][-\pl][] = \frac{1}{\ti+1} \sum_{k=0}^{\ti} \act[k][-\pl]$.

A famous variation of this procedure is smooth fictitious play (\citet{fudenbergConsistencyCautiousFictitious1995}) where players choose their action according to a regularized payoff function. Formally, a player $i$ draws an action according to a distribution that maximizes $\poff[\cdot, \auxact[\ti][-\pl][]][\pl][] + \smness \tble[\cdot, \auxact[\ti][-\pl][]]$ (with $\tble[]$ and $\smness$ defined in the previous section). An interesting property of such a procedure is that with suitable property on $\smness$, it has no \emph{regret} (up to $\smness$), meaning that if played unilaterally by a player $\pl$, other players can not trick $\pl$ into using a suboptimal action. This is due to the randomness of the action choice: the distribution assigns positive probability to every action, so a player's behavior remains unpredictable.

In this section, we extend the smooth FP procedure to stochastic games. This builds upon the recent extension of FP in \cite{leslieBestresponseDynamicsZerosum2020,sayinFictitiousPlayZerosum2020,baudinFictitiousPlayBestResponse2022} (see Section~\ref{sec:related} for more details). Indeed, it is not easy to derive the convergence to a regularized equilibrium by applying directly the definition of smooth fictitious play to the discounted stochastic game in which the players are restricted to play in stationary strategies, because the payoff function in this game is non-linear (nor is it concave, or quasi-concave) with respect to a player stationary strategy. To overcome this difficulty, and following the idea in \citet{leslieBestresponseDynamicsZerosum2020}, we update two sets of variables for every state: one concerns the uncorrelated empirical actions and the other is an estimate of the continuation payoffs (\emph{i.e.,} payoffs that players can anticipate to achieve if that state is reached). These continuation payoffs are used as a parameter in an auxiliary game, often called the Shapley operator. %Then, in our procedure, players choose actions that are the smooth best-response in the auxiliary game.

\paragraph{Auxiliary game} Following numerous authors including \citet{shapleyStochasticGames1953}, we define a so-called auxiliary game parameterized by a family of vector $\auxval[][\pl][] \in \R^\states$ (one for every player $\pl$) and a state $\st$. It is a one-shot game every player $\pl$ has the same action set as in $G$ but gets a one-time payoff of:
\begin{equation}
  \auxpoff[][\pl][\purep] := (1-\disc)\poff[\purep] + \delta \sum_{s' \in \states} \trans[\st][s'][\purep] \auxval[][\pl][s']
\end{equation}

\paragraph{Smooth Best-Response} While playing the stochastic games, players maintain a set of continuation payoffs $\auxval[][\pl][]$ that are used to choose their actions. Given that the current state is $\st$, for a player $\pl$, its action is drawn from the distribution that is the smooth best-response in the auxiliary game with respect to empirical action profile $\prof$, that is:
\begin{equation}\label{eq:sfp:sbr}
 \sbr[\st][\auxval[][\pl][]][\prof] := \arg \max_{y^\pl\in \simplex{\actions[\pl]}} \auxpoff[][\pl][y^\pl, \prof[-\pl]] + \smness \tble[y^\pl, \prof[-\pl]]
\end{equation}
This is well and uniquely-defined because of the strict concavity of $\tble[]$.

\begin{example*}
If the regularizer $\tble[]$ is taken to be the Shannon entropy, that is:
 $$\tble[\prof] = -\sum_{\pl \in \players} \sum_{\purep[\pl][] \in \actions[\pl]}\prof[\pl](\purep[\pl][]) \log \left( \prof[\pl](\purep[\pl][])\right)$$
 Then the smooth best-response function is the logit function:
 $$\sbr[\st][\auxval[][\pl][]][\prof](\purep[\pl][]) := \frac{\exp(\smness^{-1} \auxpoff[][\pl][\purep[\pl][], \prof[-\pl]])}{\sum_{\purepb[\pl][] \in \actions[\pl]}\exp(\smness^{-1} \auxpoff[][\pl][\purepb[\pl][], \prof[-\pl]])}$$
\end{example*}

\paragraph{Smooth fictitious play with known transitions and deterministic payoff}
To extend smooth fictitious play to stochastic games, we use two sets of variables. The $\auxact$ variable is the distribution of empirical actions for every state $\st$ prior to step $\ti$. The other variable $\auxval[\nti]$ can be interpreted as the continuation payoffs at $s$ and is defined as the time average of the regularized payoffs up to step $\ti$. This leads to the following system:
\begin{equation*}\label{eq:sumsfp}
 \left\{
\begin{aligned}
 & \auxval[\nti] = \frac{1}{\nti}\sum_{k=0}^{\ti}\left(\auxpoff[k] + \smness \tble[\auxact[k]]\right) \\
      & \auxact = \frac{1}{\stcount}\sum_{k=0}^{\ti} 1_{\st=\curst[k]}\act[k] \\
      & \act[\nti][\pl] \sim \sbr
    \end{aligned}
  \right.
\end{equation*}
where $\stcount$ is the number of times $\st$ was reached, that is $\sum_{k=0}^{\ti} 1_{\st=\curst[k]}$ and every action $\act[k]$ is embedded into the Euclidean space containing $\simplex{\actions[\pl]}$. The interpretation is simple: at each stage $t$, each player best replies to the belief that the other players will play according to the past uncorrelated empirical distribution of actions, and that the future continuation payoffs are equal to the time average of the past estimated perturbed payoffs, calculated using the past empirical frequencies of actions.

This system can be generalized and rewritten in an incremental fashion. It is our first main system:
\begin{equation}\label{eq:sfp}\tag{SFP}
  \left\{
    \begin{aligned}
      & \dotauxval = \rateval \left(\auxpoff + \smness \tble - \auxval\right) \\
      & \dotauxact = \rateact \left(\sbr - \auxact\right)
    \end{aligned}
  \right.
\end{equation}
where $\speedval=1$ corresponds to the equation above.

\begin{remark*}
If there is only one state (e.g the classical repeated game sitting), this procedure is exactly standard smooth fictitious play since the smooth best-response does not depend on the value of $\auxval$. 
\end{remark*}

\paragraph{Smooth fictitious play with unknown transitions and perturbed payoff}

We suppose now that payoff functions are not known, and that each stage payoff is observed with some zero-mean noise which follows a distribution that may depend on the history, the current state and actions taken by the other players. Therefore, at step $\ti$, player $\pl$ gets a random reward $\tipoff$ that is drawn according to a distribution determined by actions $\act$ and current state $\curst$ whose expectancy is $\poff[\act][\pl][\curst]$ and bounded variance conditionally on the history (see Appendix~\ref{app:sfp} for details). We also suppose that transitions are not known. Therefore, both transitions and expected payoff may be empirically estimated as follows:
\begin{equation}\label{eq:estimate}\tag{MFP.1}
 \left\{
 \begin{aligned}
 & \esttrans[\st][s'][a] = \frac{\sum_{k=0}^{\ti} 1_{\curst[k]=\st\land \act[k]=a} 1_{\curst[k+1]=s'} }{\sum_{k=0}^{\ti} 1_{\curst[k]=\st\land \act[k]=a}}\\
 & \estpoff[a] = \frac{\sum_{k=0}^{\ti} 1_{\curst[k]=\st\land \act[k]=a} \tipoff[k] }{\sum_{k=0}^{\ti} 1_{\curst[k]=\st\land \act[k]=a}}\\
\end{aligned}
 \right.
\end{equation}
Consequently, we define the estimated auxiliary payoff using these two estimators:
\begin{equation}\label{eq:sfp:estauxpoff}\tag{MFP.2}
 \estauxpoff[][\pl][\prof] := (1-\disc)\estpoff[\prof] + \delta \sum_{s' \in \states} \esttrans[\st][s'][\prof] \auxval[][\pl][s']
\end{equation}
Now we define a model-free version of smooth fictitious play, similar to \ref{eq:sfp} but using estimators:
\begin{equation}\label{eq:mfp}\tag{MFP}
 \left\{
 \begin{aligned}
& \dotauxval = \rateval \left(\estauxpoff + \smness \tble - \auxval\right) \\
 & \dotauxact = \rateact \left(\sbr - \auxact\right)
\end{aligned}
 \right.
\end{equation}
where $\sbr[\st][\cdot][]$ is defined relatively to $\estauxpoff[][\pl][][\st][\cdot]$.

\begin{theorem}[Convergence in identical-interest ergodic DSG]\label{thm:sfp:team}
In a identical-interest ergodic discounted stochastic game, if all players follow \ref{eq:sfp} (resp. \ref{eq:mfp}), their empirical actions $\auxact$ converge almost surely to the set of regularized stationary Nash equilibria and their expected vector of continuation payoffs $\auxval$ converges to the optimal continuation payoff of limiting equilibrium set.
\end{theorem}

This theorem means that even if the trajectories of \ref{eq:sfp} (resp. \ref{eq:mfp}) cycle between several stationary equilibria, they all share the same optimal continuation payoff vector.

\begin{theorem}[Convergence in zero-sum ergodic DSG]\label{thm:sfp:zerosum}
In a zero-sum ergodic discounted stochastic game, if all the players follow \ref{eq:sfp} (resp. \ref{eq:mfp}) with the same initial values, their empirical actions $\auxact$ converge almost surely to the set of $D\speedval$-regularized Nash equilibria (where $D> 0$ is a constant that only depends on $G$) and their expected vector of continuation payoffs $\auxval$ converges to the corresponding continuation payoff.
\end{theorem}

To prove these results, we are going to define in the next section an associated smooth best-response continuous time counterpart to our independent learning algorithms. The proof of Theorems~\ref{thm:sfp:team} and~\ref{thm:sfp:zerosum} are sketched at the end of Section~\ref{sec:sbrd} and fully detailed in Appendix~\ref{app:sfp}.

\begin{remark}Theorem~\ref{thm:sfp:zerosum} suggests that it is possible to use a doubling-trick mechanism to converge to the set of 0-regularized stationary Nash equilibria. Indeed, players can compute the duality gap (see Appendix~\ref{app:sbr} for a definition) and decide to reduce the update rate $\beta$ every time it is below a certain threshold. 
This is a standard trick to achieve no-regret in reinforcement learning.
\end{remark}

\begin{remark}
Our rules suppose that each player observes the other player's past actions. This is an important difference with the algorithms developed in \citet{NEURIPS2020_3b2acfe2,LPOP2022} where a player observes only its own actions. On the other hand, they do not prove a time-average convergence of their trajectories, while we do, but they prove a best iterate convergence that we do not prove. %Designing a decentralized algorithm where the time-average of its trajectiories converges to stationary Nash equilibria in zero-sum or identical interests DSG where the players observe only the states, their own actions and their stage payoffs is an open problem.
\end{remark}

}

\section{Smooth Best-Response in Stochastic Games  (Continuous-Time Dynamics)}\label{sec:sbrd}

\paragraph{Continuous counterpart of discrete time systems} 
%It is often easier to prove the convergence of continuous time systems. 
The continuous-time counterpart of fictitious play is best-response dynamics of \cite{matsuiBestResponseDynamics1992,harrisRateConvergenceContinuousTime1998a}. Together with the theory of stochastic approximations (\cite{benaimStochasticApproximationsDifferential2005}), proofs of convergence of best-response are key ingredients to modern proofs of convergence of fictitious play and smooth FP in repeated games (\cite{benaimStochasticApproximationsDifferential2005, benaimConsistencyVanishinglySmooth2013}) and more recently in proving the convergence of FP in stochastic games (\cite{sayinFictitiousPlayZerosum2020,baudinFictitiousPlayBestResponse2022}).

In this section, we define the continuous-time counterpart to our discrete-time algorithms. This leads us to a regularized version of the continuous-time best-response dynamics in \citet{leslieBestresponseDynamicsZerosum2020,baudinFictitiousPlayBestResponse2022}. Similarly to the previous section, our continuous time system has two sets of variables $\auxval$ which represents the belief about the continuation payoffs and $\auxact$ which corresponds to belief about the players actions. Our main continuous-time system is:
\begin{equation}\left\{
  \begin{aligned}
    \dotauxval =\ & \rateval \left(\auxpoff + \smness \tble -\auxval\right) \\
    \dotauxact =\ & \rateact \left(\sbr - \auxact\right) \\
    \rateact & \in [\ratemin, 1]
   % \dot{x}_s^i(t) = \ratest st \left( \sbestr -\cxs i t\right) \\
   % \ratest s t \in [\ratelb, 1]
  \end{aligned}
\right.\tag{SBRD}\label{eq:sbrd}\end{equation}
% where $\sbestr := \arg \max_{y \in \Delta(A^i)} \stfs {\cut}{y^i, \cxs{-i} t} + \e \h[y^i, \cxs{-i}{t}]$.
\begin{remark}
  In zero-sum or identical-interest games, if players have the same initial conditions, then we can omit the superscript $\pl$ in $\auxval[]$ as they are equal.
\end{remark}

\paragraph{Update rates}
Profile $\auxact$ evolve towards the smooth best-response at a rate of $\rateact$. Variable $\rateact$ corresponds to the frequency at which a state $\st$ is visited. The fact that it is bounded below by $\ratemin > 0$ is a mathematically convenient way to exploit the ergodicity of the stochastic game.

\begin{remark}
  Note that this is different from the model of continuous-time stochastic game outlined for instance by \cite{neymanContinuoustimeStochasticGames2017}. In this paper, we are interested in using the theory of stochastic approximation, therefore discrete-time and continuous-time system are related through an exponential change of variable in time. As a consequence, we can consider that the state occupation is averaged, which is not the case in some other studies in stochastic games.
\end{remark}

Regarding continuation payoffs, they evolve towards perturbed payoffs at a rate of $\rateval$, which is typically oblivious, that is only determined by time $\ti$ and state $\st$ independently of the value of other variables. We make the following assumptions that guarantee that the update rates are not too small:
\begin{equation}
  \tag{H2}
  \label{hyp:ratesbr}
  \begin{gathered}
    \rateval \geq 0 \text{\ and\ } \rateval[]\text{\ is decreasing} \\
    \int_{0}^\infty \rateval[u]du = +\infty
  \end{gathered}
\end{equation}

\begin{remark}
  With $\rateval = \frac{1}{t+1}$, the system is close to the one outlined in \cite{leslieBestresponseDynamicsZerosum2020}. In this case, continuation payoffs are updated with a \emph{slower} timescale. Indeed, having two timescales may be useful to establish convergence or apply stochastic approximation theorems. In our paper, we do not make such an assumption, allowing the possibility that the two variables are on a single timescale.
\end{remark}

\paragraph{Model-Free System} In order to also study the convergence of the model-free version of our procedure, we also define the smooth best-response dynamics when the model is progressively learned. The estimators are defined as follows:
\begin{equation}\label{eq:cont:estimate}
  \left\{
    \begin{aligned}
      & \dotesttrans[\st][s'][b] = \rateact \act (b) \left( \trans[\st][s'][b] - \esttrans[\st][s'][b]\right)\\
      & \dotestpoff[b] = \rateact \act (b) \left( \poff[b] - \estpoff[b]\right)\\
      & \act[\ti][\pl] = \sbr \\ %[\st][\auxval[\ti][][]][\auxact][]
    \end{aligned}
  \right.
\end{equation}
where $\act := \Pi_{\pl \in \players}\act[\ti][\pl]$ is the profile of selected actions.

\NewDocumentCommand{\refmbrd}{}{MBRD{}\xspace}

Then, $\estauxpoff[\ti][\pl][][\st][\cdot]$ is defined as in (\ref{eq:sfp:estauxpoff}) and we obtain a system \refmbrd by using estimators in \ref{eq:sbrd}.

\paragraph{Solutions}

Systems \ref{eq:sbrd} and \refmbrd are differential inclusions. A general theory can be found in \cite{aubinDifferentialInclusionsSetValued1984}. Here, the right-hand side of each system forms a closed, set-valued map (because there are several possible values for $\rateact$) whose images are convex and uniformly bounded. Under these assumptions, it is known that such systems admit (typically non unique) solutions.

%\paragraph{Convergence} We are interested in the asymptotic behavior of such solutions for two reasons. First, for their own sake, this dynamics has been proven to converge in the non-stochastic and/or non-smooth cases. They will also allow proving results on discrete time systems introduced in Section~\ref{sec:sfp}. 

\begin{theorem}\label{thm:sbr:cv}
  Under hypothesis~\ref{hyp:ratesbr}, \refmbrd converges to the set of regularized Nash equilibria in identical-interest stochastic games. Furthermore, if $\ratelim \geq \lim\sup_{\ti \rightarrow \infty} \rateval$, then \refmbrd converges to the set of $D \ratelim$-regularized stationary Nash equilibria in zero-sum games, where $D$ is a positive constant that depends only on the game $G$. In particular if $\rateval$ goes to $0$ then \refmbrd converges to the set of regularized stationary Nash equilibria.
\end{theorem}

The convergence of \refmbrd is helpful to characterize the limit set of discrete-time smooth fictitious play systems. Indeed, they are contained in the internally chain transitive sets (see Appendix~\ref{app:sfp} for a definition) of \refmbrd using the theory of stochastic approximations. Below, we sketch the proof of continuous-time Theorem~\ref{thm:sbr:cv}. The discrete-time counterpart Theorem~\ref{thm:sfp:team} and Theorem~\ref{thm:sfp:zerosum} can be deduced in a similar fashion as \cite{baudinFictitiousPlayBestResponse2022}, although there are several technical subtleties. Complete proofs of these results are in Appendix~\ref{app:sbr}.

\begin{proof}[Sketch of the proof of Theorem~\ref{thm:sbr:cv}]

The proofs for both class of games proceed quite differently: this is not surprising since there are no (to the best of our knowledge) unified convergence proof of even simple FP in potential and zero-sum games.

For identical-interest stochastic games, the key point is showing that the gap between (estimated) auxiliary payoffs $\auxpoff$ and the auxiliary values $\auxval$ is narrowing. Technically, the difference is bounded below by an integrable function, which implies that $\auxval$ converges and then that $\auxpoff$ converges and their limits are necessarily equal. Then, a study of the behavior of actions shows that they belong to the set of regularized equilibria, otherwise $\auxpoff$ could not converge (based on Lipschitzian properties of all these quantities).

Regarding zero-sum stochastic games, the first part of the proof studies a rather standard quantity called the duality gap. It goes to $0$ (or at least near $0$), which implies that the value of the auxiliary game is mostly reached by the auxiliary payoffs $\auxpoff$. Then, comparing relative speed of auxiliary values in all states leads to the convergence of these values and of other variables.
  
\end{proof}

\iffalse

\begin{proof}[Sketch of the proof of Theorem~\ref{thm:sfp:team}]
  
\end{proof}

\begin{proof}[Sketch of the proof of Theorem~\ref{thm:sfp:zerosum}]
  
\end{proof}

\fi

\ifexample
\section{Example: Soccer Game}\label{sec:soccer}

\fi

\section{More on the Related Works}\label{sec:related}

\paragraph{Single-Player Markov Decision Process}
\citet{watkinsLearningDelayedRewards1989} introduced Q-learning to solve Markov Decision Process with guarantees when the environment is stationary. The central idea is to estimate the so-called Q-values for every state-action pairs that represents the continuation payoff using the following scheme:
\begin{equation}\label{eq:related:qlearning}
    Q(s, a) \leftarrow Q(s, a) + \beta(r_s(a) + \delta \max_{a'} Q(s', a') - Q(s, a))
\end{equation}
where $s$ is the current state, $a$ the action played, $s'$ the next state. Q-values for other state-action pairs are left unchanged. This simple rule does not need any information about the transitions. 

% In particular continuation values and the transition functions are estimated simultaneously.

Since then, numerous extensions have been introduced, for instance double Q-learning (\cite{NIPS2010_091d584f}) or Q-learning with deep learning (\cite{hasseltDeepReinforcementLearning2016}). However, in multiagent systems, from the one player's point of view the environment is not stationary because it comprises other players that are also using some non stationary learning rule. Therefore, there is no guarantee and indeed it may not converge to the set of Nash equilibria (\cite{wunderClassesMultiagentQlearning2010,calvanoArtificialIntelligenceAlgorithmic2020}. Our paper is a way to overcome this difficulty with a different updating rule for estimating the continuation payoffs. This updating rule is closer to Expected Sarsa (\cite{suttonReinforcementLearningIntroduction2018}) where the continuation values are moved towards the \emph{expectation} of future payoffs.

%Below, we  describe other extensions of Q-learning to multiagent systems.

%\paragraph{Decentralized Q-learning}
%\citet{sayinDecentralizedQLearningZerosum2021} introduced a decentralized version of Q-learning using a two-timescale stochastic approximation: players estimate at the same time a Q-function and a set of values for every state (which corresponds to $max_{a'} Q(s', a')$ of (\ref{eq:related:qlearning})). Compared to our work, 
%this does not model the other player as in fictitious play. Furthermore,  it uses a two-timescale scheme that we do not require, and is proved to converge only in zero-sum SG.

\paragraph{Nash Q-learning} \cite{huNashQlearningGeneralsum2003} is another way to learn Nash equilibria based on Q-learning, supposing that players %can solve the auxiliary game parameterized by the Q-values and 
have access to a Nash equilibrium of the auxiliary game parameterized by the Q-values. In our paper, we do not suppose that the players have access to such an oracle.

\paragraph{Repeated games}
When there is only one state, this is the classical framework in which the players repeatedly play a fixed normal form game. \citet{leslieIndividualQLearningNormal2005} studies Q-learning in this context.
On the FP and SFP side, most of the literature assumes that the players know precisely the parameters of the problems. We only assume that the stage rewards are imperfectly observed. Hence, our results also improve upon the classic framework.
%\footnote{In \citet{hofbauerGlobalConvergenceStochastic2002} stochastic FP model, payoffs are perturbed at every stage and the one received are the real material payoffs, but the players know the mean payoff function and know the law of the perturbation and thus, are able to compute a best response to that perturbed game (which turns out to be equivalent to SFP and regularized FP). This is different from our framework where the players do not know the law of the noise nor the mean payoff function.}

\paragraph{Fictitious Play Algorithms in Stochastic Games} \citet{brown1951iterative} and \citet{robinsonIterativeMethodSolving1951} introduced fictitious play to learn the value of repeated fixed zero-sum game in a turn based fashion (see also \citet{bergerBrownOriginalFictitious2007}. It was later studied in the usual simultaneous updating by \cite{fudenbergTheoryLearningGames1998} with numerous extensions, including smooth fictitious play. However, none of these procedures converge to equilibria in the repetition of the zero-sum normal form game $\Gamma$ induced by the stationary strategies of our DSG, mainly because the payoff functions of $\Gamma$, on a player own-strategy, are neither linear, nor concave \cite{daskalakisComplexityConstrainedMinmax2021} and so new idea are necessary to overcome this difficulty. 

%Therefore, new ideas are required to define procedures in stochastic games and prove their convergence.

\citet{vriezeFictitiousPlayApplied1982} studied fictitious play in the context of varying stage games. However, similarly to the original fictitious play of \citeauthor{robinsonIterativeMethodSolving1951} and \citeauthor{robinsonIterativeMethodSolving1951}, it is designed as a way to compute the value of a zero-sum game, and it is not a behavioral strategy (i.e., there is no current state, the goal is to compute the minmax of a series of matrices). More recently, \citet{leslieBestresponseDynamicsZerosum2020} proposed a best-response dynamics in continuous time which converges to stationary Nash equilibria in zero-sum DSG. The convergence was extended to discrete-time procedures and to identical-interest discounted stochastic games by \citet{baudinFictitiousPlayBestResponse2022}. These discrete-time and continuous-time systems are close to our systems, but they only work in model-based settings with no perturbation, whereas ours is suited for model-free settings and uses smooth best-responses.  \citet{sayinFictitiousPlayZerosum2020} introduced a fictitious play algorithm for zero-sum stochastic games in discrete time with a continuation up-dating rule close to Q-learning: both players update a Q-table with an entry per state-action pair whereas our continuation payoff are only indexed by the state. It works in a model-free setting, but with exact best-responses and it was proved to converge only in zero-sum SG.

\paragraph{Smooth or Regularized Learning} Learning using regularizers is a widespread technique in machine learning, for instance with "follow the perturbed leader" (\cite{cesa-bianchiPredictionLearningGames2006}) or so-called stochastic fictitious play (\cite{fudenbergConsistencyCautiousFictitious1995}). It makes it possible to have no-regret (\cite{perchetApproachabilityRegretCalibration2014a,shalev-shwartzOnlineLearningOnline2011a}) and from another point of view it is a substitute to the simpler $\epsilon$-greedy exploration scheme in reinforcement learning (\cite{suttonReinforcementLearningIntroduction2018}).

\paragraph{Local vs Global} Our algorithms lead to all the players globally optimizing. Other algorithms such as COMA (\cite{foersterCounterfactualMultiagentPolicy2018}) lead to local optima, and so belong to another line of work.

\section{Conclusion}

We defined a number of decentralized continuous and discrete time systems with exploration to learn its own payoff function, the transition probability function, and which converge to regularized (and so $\varepsilon$-approximate) stationary Nash equilibria in discounted stochastic games. This is an area that has not been widely studied and as such, a number of questions remain to be answered--either regarding our systems or regarding other recent decentralized learning algorithms in stochastic games.

First, the theory of stochastic approximations gives no clue about the speed of convergence of the corresponding discrete-time algorithms. Moreover, even in continuous time, this question remains open even for potential games \cite{harrisRateConvergenceContinuousTime1998a}. Therefore, an interesting question in both the stochastic and the repeated game setting is, which guarantee FP and SFP have regarding the rate of convergence?

Second, and inspired by the convergence of vanishing fictitious play to equilibria in classical repeated games (see \citet{benaimConsistencyVanishinglySmooth2013,hadikhanlooLearningNonatomicGames2021}), it would be of interest to study a vanishing version of our algorithm with a parameter $\varepsilon(t)$ that goes to zero in a way that guarantees the convergence to (exact) stationary Nash equilibria of the DSG. 

Third, convergenve of FP or SFP in repeated game extends to infinite action games and to non atomic games. It would be of interest to relaxe our finiteness assumption, at least for the action sets.

A challenging open problem is the design of independent learning algorithms that converge to stationary equilibria in ergodic zero-sum and identical-interest DSG without the knowledge of the other player's past actions. The  Projected Gradient Methods of \citet{NEURIPS2020_3b2acfe2,LPOP2022} have this minimal information property but their results are proved only episodic SG.

Last but not least, obtaining a last iterate convergence instead of a time average convergence (as in our paper) or best iterate convergence (as in \citet{NEURIPS2020_3b2acfe2,LPOP2022}) is another interesting direction.

\bibliography{these}

\newpage

\appendix

\section{Convergence of smooth best-response dynamics in identical-interest and zero-sum stochastic games}\label{app:sbr}

In this appendix, we provide detailed proofs of results of Section~\ref{sec:sbrd}. Since solutions of differential inclusion~\ref{eq:sbrd} are special solutions of differential inclusion~\refmbrd, it is sufficient to study \refmbrd. In what follows, unless otherwise specified, $\ti \mapsto (\auxact[], \auxval[][\pl], \esttrans[\st][][][], \estpoff[][\pl][\st][])_{\st, \pl}$ is a solution of \refmbrd.

We start with general properties of the smooth best response, convergence of estimates and regularity of solutions and then proceed with the study of two special cases: identical-interest stochastic games (subsection~\ref{app:sbr:ii}) and zero-sum stochastic games (subsection~\ref{app:sbr:zs}).

\subsection{Properties of the smooth best-response}

\NewDocumentCommand{\minsbr}{}{\eta}
\NewDocumentCommand{\poffbound}{}{B}

{
    \RenewDocumentCommand{\auxval}{O{} O{} O{}}{u}
    \RenewDocumentCommand{\auxact}{O{} O{} O{\st}}{x_{#3}}
    \RenewDocumentCommand{\act}{}{\purep[\pl][\st]}

    \begin{lemma}
        \label{lem:sbr:cont}
        Function $(\auxval, \auxact) \mapsto \sbr$ is continuous in $\auxval$ and $\auxact$.
    \end{lemma}

    \begin{proof}
        It follows from a simple application of the maximum theorem that the smooth best-response is upper hemicontinuous as a set-valued map. The strict concavity of $\tble[]$ implies that there is a single profile that maximizes the smooth auxiliary payoff. Therefore $\tble[]$ is continuous as a single-valued application.
    \end{proof}

    \begin{lemma}\label{lem:sbr:minsbr}
        For $\poffbound > 0$, there exists $\minsbr > 0$ such that for all $x \in \Pi_{j\in \players} \simplex{\actions[j]}$, $\auxval \in \R^\states$ such that $\|u\|_{\infty} \leq \poffbound$ and $\act \in \actions[\pl]$:
        $$\sbr[\st](\act) \geq \minsbr$$
    \end{lemma}
    \begin{proof}
        This is a classical property of the smooth best-response under assumptions~(\ref{eq:tble}).  Point $\sbr$ is a maximum of the function $\profb[\pl] \mapsto \auxpoff[][\pl][\profb[\pl], \prof[-\pl]] + \smness \tble[\profb[\pl], \prof[-\pl]]$.
        However, let $\prof$ be an interior point of the simplex and $(\profb[\pl], \prof[-\pl])$ be on the boundary for player $\pl$. Then, the composition of the linear interpolation between $(\profb[\pl], \prof[-\pl])$ and $\prof$ and $\prof \mapsto \auxpoff[][\pl][\prof] + \smness \tble[\prof]$ is concave because both functions are. However, the slope of this composed function is infinite (because the norm of the gradient goes to $\infty$ and $\prof$ is interior), therefore it goes to $-\infty$ (otherwise it cannot be concave), which implies that $(\profb[\pl], \prof[-\pl])$ cannot be a maximum. Then, by compacity, smooth best-response are away from the boundary of the simplex.
    \end{proof}

}

\subsection{Convergence of estimates $\esttrans[][][][]$ and $\estpoff[][][][]$}

    \NewDocumentCommand{\cvestbound}{}{C}

    \begin{lemma}\label{lem:sbr:cvest}
        There exists $\cvestbound > 0$ such that for all states $\st$ action profiles $\purepb \in \actions$ and $\ti \geq 0$,
        $$|\estpoff[\purepb] - \poff[\purepb]| \leq \cvestbound \exp(-\minsbr \ratemin \ti)$$
        $$|\esttrans[\st][\stb][\purepb] - \trans[\st][\stb][\purepb]| \leq \cvestbound \exp(-\minsbr \ratemin \ti)$$
    \end{lemma}
    \begin{proof}{
        \NewDocumentCommand{\f}{O{\ti}}{r\braced{#1}}
        Let $\st \in \states$ and $\purepb \in \actions$.\newline We write $\f = \left|\estpoff[\purepb] - \poff[\purepb]\right|$, so as $\dot {\f[]} = -\rateact \act(\purepb)\f$.

        Then we can deduce from Lemma~\ref{lem:sbr:minsbr} and the fact that $\rateact\geq \ratemin$ that $\dot {\f[]} \leq - \minsbr \ratemin \f$.

        Therefore, Grönwall Lemma implies that $\f \leq \f[0] \exp(-\minsbr \ratemin \ti)$. The same proof is valid for~$\esttrans$.
    }\end{proof}

\subsection{Regularity of solutions}
    \begin{lemma}\label{lem:sbr:bounded}
        Functions $\auxval[][\pl]$, $\auxact[]$, $\estauxpoff[][\pl][]$ are bounded.
    \end{lemma}

    \begin{proof}

        Because of the definition of the derivative of $\auxact[]$, it stays in the simplex and as such it is bounded.

        It is clear that $\estauxpoff \leq (1-\disc)\|\poff[]\|_\infty + \disc \|\auxval[\ti][][\pl]\|_\infty$. Therefore, as long as $\|\auxval[\ti][\pl][]\|_\infty$ is lower than $\|\poff[]\|_\infty$, then $\estauxpoff \leq \|\poff[]\|_\infty$. By definition of the derivative of $\auxval[][\pl]$, this implies that $\auxval[][\pl]$ is always smaller than $\|\poff[]\|_\infty$ assuming it is true for the initial value.
    \end{proof}

    \begin{lemma}\label{lem:sbr:lip}
        Functions $\auxval[]$, $\auxact[]$, $\estauxpoff[][\pl][]$ are Lipsichitz.
    \end{lemma}
    \begin{proof}
        All functions are differentiable almost everywhere. The derivatives are bounded by composition since $\rateval$ and $\rateact$ are bounded and because of Lemma~\ref{lem:sbr:bounded}.
    \end{proof}

\subsection{Convergence in identical-interest games}\label{app:sbr:ii}{
    \NewDocumentCommand{\abbrauxpoff}{O{\ti} O{} O{} O{\st} O{}}{\Gamma_{#4}\braced{#1}}
    \NewDocumentCommand{\stm}{O{\ti}}{s_-\braced{#1}}
    \NewDocumentCommand{\argmin}{m}{\arg\min_{#1}}
    In identical-interest games, all payoff functions are equal. Therefore, assuming the same initial conditions for all players $\pl$ for $\auxval[][\pl]$, we have $\auxval[\ti][\pl]=\auxval[\ti][\plb]$ for all players $\pl$ and $\plb$. So, we can omit the superscript, and the same is true for the payoff functions and the auxiliary payoff functions.
    
    The proof proceeds as follows: we show that the differential of $\auxval[]$ becomes a good approximation of the target auxiliary payoff $\abbrauxpoff[]$, and furthermore that their difference is lower bounded by something integrable. Then we can deduce that there is a limit, for this difference and that it is necessarily $0$ and the convergence of actions follows.

\RenewDocumentCommand{\estauxpoff}{O{\ti} O{\pl} O{\auxact[#1]} O{\st} O{\auxval[#1][][]}}{{\hat f}_{#4, #5}\braced{#3}}

\RenewDocumentCommand{\auxpoff}{O{\ti} O{\pl} O{\auxact[#1]} O{\st} O{\auxval[#1][#2][]}}{f_{#4, #5}\braced{#3}}

\RenewDocumentCommand{\poff}{O{} O{\pl} O{\st}}{r_{#3}\braced{#1}}

\RenewDocumentCommand{\estpoff}{O{\act} O{\pl} O{\st} O{\ti}}{\hat r_{#3}\braced{#1}\braced{#4}}

\RenewDocumentCommand{\auxval}{O{\ti} O{} O{\st}}{u_{#3}\sbraced{#1} }
\RenewDocumentCommand{\dotauxval}{O{} O{} O{\st}}{\dot{u}_{#3} \ifthenelse{\equal{#1}{}}{}{\left( #1 \right)}}

\RenewDocumentCommand{\tble}{O{\auxact} O{}}{h^{#2}\braced{#1}}
    \subsubsection{Convergence of payoffs}

    We define $\abbrauxpoff := \estauxpoff + \smness \tble$ and $\stm \in \argmin{\st \in \states} \abbrauxpoff - \auxval$. For a given solution of \refmbrd, there may be several possible choices of $\stm$, results below are valid for any such choice.

    \RenewDocumentCommand{\auxpoff}{}{\abbrauxpoff}
    \NewDocumentCommand{\mauxval}{O{\ti} O{#1} O{\stm[#2]}}{\auxval[#1][][#3]}
    \NewDocumentCommand{\mauxpoff}{O{\ti} O{#1} O{\stm[#2]}}{\auxpoff[#1][][][#3]}

    \begin{lemma}\label{lem:sbr:inegdiff} The differential of $\mauxpoff[]-\mauxval[]$ is lower bounded for almost every $\ti$:
        $$\diff{\mauxpoff[]-\mauxval[]} \geq -2\cvestbound \exp(-\minsbr\ratemin\ti) + \rateval (\disc-1)(\mauxpoff-\mauxval)$$
        where $\cvestbound$ is defined in Lemma~\ref{lem:sbr:cvest} and $\minsbr$ in Lemma~\ref{lem:sbr:minsbr}.
    \end{lemma}
    \begin{proof}{
        \NewDocumentCommand{\h}{}{h}
        First, notice that since $\mauxpoff[]-\mauxval[]$ is a minimum of continuous, differentiable almost everywhere functions, it is continuous and differentiable almost everywhere. Let $\ti \in \R^+, \h > 0$.

        \begin{equation}\label{pr:sbr:diff0}
            \begin{aligned}
            \mauxpoff[\ti+\h]-\mauxval[\ti+\h] - \mauxpoff[\ti] + \mauxval[\ti] \\
             = \mauxpoff[\ti][\ti+\h] + \h\diff{\mauxpoff[][\ti+\h]}(\ti)-\mauxval[\ti][\ti+\h]\\-\h\diff{\mauxval[][\ti+\h]}(\ti) + o(\h) - \mauxpoff[\ti] + \mauxval[\ti]
            \end{aligned}
        \end{equation}

        Then, for any $\h$:
        \begin{equation}\label{pr:sbr:diff1}
            \begin{aligned}
            \diff{\mauxval[][\ti+\h]}(\ti) & = \rateval\left(\auxpoff[\ti][][][\stm[\ti+\h]] - \auxpoff[\ti][][][\stm[\ti+\h]]\right) \\ & = \rateval \left(\mauxpoff-\mauxval\right) + o(1)
            \end{aligned}
        \end{equation}

        Moreover, for any $\h$:
        \begin{equation}\label{pr:sbr:diff2}
            \begin{aligned}
            \mauxpoff[\ti][\ti+\h] - \mauxval[\ti][\ti+\h] \geq \mauxpoff[\ti] - \mauxpoff[\ti]
            \end{aligned}
        \end{equation}

        Now, we need to lower bound $\diff{\mauxpoff[][\ti+\h]}(\ti)$. We use the chain rule and Lemma~\ref{lem:sbr:cvest} to get (with $\st=\stm[\ti+\h]$ to ease reading--since this differentation does not depend on the state):
        \begin{equation}\label{pr:sbr:diff5}
            \begin{aligned}
                \diff{\auxpoff[]}(\ti) \geq
                \underbrace{-2\cvestbound \exp(-\minsbr \ratemin \ti)}_{\text{changes in\ } \estpoff[][][][]\text{\ and\ }\esttrans[][][][]}
                + \underbrace{\sum_{\plb \in \players} \rateact \langle \sbr[\st][\auxval[\ti][\pl]][\auxact[\plb]][\plb]-\auxact[\ti][\plb], \nabla_\plb \estauxpoff[\ti][\pl][\auxact[\ti]]+\smness \tble\rangle}_{\text{changes in\ }\auxact[], \ \nabla_j\text{\ is the gradient with respect to\ }\auxact[][\plb]} \\
                + \underbrace{\disc \sum_{\stb \in\states} \esttrans[\st][\stb][\auxact][]\dotauxval[\ti][\pl][\stb]}_{\text{changes in\ }\auxval[][]}
            \end{aligned}
        \end{equation}
 
        The second term is positive because $\estauxpoff[][][]+\smness \tble$ is concave.\newline The last one is greater than $\disc\rateval\left(\mauxpoff-\mauxval\right)$, leading to:
        \begin{equation}\label{pr:sbr:diff3}
            \begin{aligned}
                \diff{\mauxpoff[][\ti+\h]}(\ti) \geq -2\cvestbound \exp(-\minsbr \ratemin \ti) + \disc\rateval\left(\mauxpoff-\mauxval\right)
            \end{aligned}
        \end{equation}

        Using (\ref{pr:sbr:diff1}), (\ref{pr:sbr:diff2}) and (\ref{pr:sbr:diff3}) in (\ref{pr:sbr:diff0}) leads to:
        \begin{equation*}\label{pr:sbr:diff4}
            \begin{aligned}
                & \mauxpoff[\ti+\h]-\mauxval[\ti+\h] - \mauxpoff[\ti] + \mauxval[\ti] \\
                & \geq \h \left( -2\cvestbound \exp(-\minsbr\ratemin\ti) + \rateval(\disc-1) \left(\mauxpoff-\mauxval\right)\right) + o(\h)
            \end{aligned}
        \end{equation*}

        And the result follows.

    }\end{proof}

    \NewDocumentCommand{\limval}{}{u_{\infty}}

    \begin{lemma}\label{lem:sbr:cvpoff}
        There exists $\limval \in \R^\states$ such that for all $\st \in \states$, $\auxval \rightarrow \limval$ and $\auxpoff \rightarrow \limval$.
    \end{lemma}
    \begin{proof}
        We consider the following differential equation:
        \begin{equation}
            \label{pr:sbr:prcv1}
            y' = -2\cvestbound \exp(-\minsbr\ratemin\ti) + \rateval(\disc-1) y
        \end{equation}
        Lemma~\ref{lem:sbr:inegdiff} implies that solutions of (\ref{pr:sbr:prcv1}) with initial condition $y(0) = \mauxpoff[0] - \mauxval[0]$ lower bound $\mauxpoff[]-\mauxval[]$.

        Moreover, solutions of (\ref{pr:sbr:prcv1}) are of the form:
        \begin{equation*}
            \left(y(0) - \int_0^{\ti} 2\cvestbound\exp(-\minsbr\ratemin v)\exp\left(\int_0^v \rateval[u](1-\disc) du\right) dv\right)\exp\left(-\int_0^\ti \rateval[u](1-\disc) du\right)
        \end{equation*}
        which is equal to:
        \begin{equation*}
            y(0)\exp\left(-\int_0^\ti \rateval[u](1-\disc) du\right) - \int_0^{\ti} 2\cvestbound\exp(-\minsbr\ratemin v) \exp\left(-\int_v^\ti \rateval[u](1-\disc) du\right)dv
        \end{equation*}
        Which is greater than:
        \begin{equation} \label{pr:sbr:prcv2}
            \begin{gathered}
                - |y(0)|\exp\left(-\int_0^\ti \rateval[u](1-\disc) du\right) - \int_0^{\ti} 2\cvestbound\exp(-\minsbr\ratemin v) \exp\left(-\int_v^\ti \rateval[u](1-\disc) du\right)dv
            \end{gathered}
        \end{equation}

        So for every $\st$, $\auxpoff-\auxval$ is greater than (\ref{pr:sbr:prcv2}). We study the differential of $\auxval$, that is $\rateval(\auxpoff-\auxval)$ and show that it is lower bounded by an integrable quantity:
        \begin{multline}\label{pr:sbr:prcv3}
            \dotauxval \geq - \rateval |y(0)|\exp\left(-\int_0^\ti \rateval[u](1-\disc) du\right) \\ - \rateval\int_0^{\ti} 2\cvestbound\exp(-\minsbr\ratemin v) \exp\left(-\int_v^\ti \rateval[u](1-\disc) du\right)dv
        \end{multline}
        
        The integral of the first term is:
        \begin{multline*}
            \int_0^\ti - \rateval[v] |y(0)|\exp\left(-\int_0^v \rateval[u](1-\disc) du\right) \\ = - |y(0)|\frac{1}{1-\disc}\left(1-\exp\left(-\int_0^t \rateval[u](1-\disc) du\right)\right) > -\infty
        \end{multline*}
        
        And the integral of the second term is:
        \begin{align*}
        & \int_0^\ti - \rateval[v]\int_0^{v} 2\cvestbound\exp(-\minsbr\ratemin w) \exp\left(-\int_w^v \rateval[u](1-\disc) du\right)dw dv \\
        = &\int_0^\ti \int_0^\ti -1_{w \leq v}\rateval[v]2\cvestbound\exp(-\minsbr\ratemin w) \exp\left(-\int_w^v \rateval[u](1-\disc) du\right)dw dv \\
        = &\int_0^\ti \int_w^\ti -\rateval[v]2\cvestbound\exp(-\minsbr\ratemin w) \exp\left(-\int_w^v \rateval[u](1-\disc) du\right)dv dw \\
        = &\int_0^\ti 2\cvestbound\exp(-\minsbr\ratemin w) \int_w^\ti -\rateval[v] \exp\left(-\int_w^v \rateval[u](1-\disc) du\right)dv dw \\
        = &\int_0^\ti 2\cvestbound\exp(-\minsbr\ratemin w) \left[\frac{1}{1-\delta} \exp\left(-\int_w^v \rateval[u](1-\disc) du\right)\right]_w^\ti dw\\
        = &\int_0^\ti 2\cvestbound\exp(-\minsbr\ratemin w) \frac{1}{1-\delta}\left(\exp\left(-\int_w^t \rateval[u](1-\disc) du\right)-1\right) dw \\
        \geq & - \frac{1}{1-\disc}\frac{2\cvestbound}{\minsbr\ratemin}(1-\exp(-\minsbr\ratemin \ti)) > - \infty
        \end{align*}
        Therefore, both term of the integral of (\ref{pr:sbr:prcv3}) are greater than $-\infty$. Since $\auxval$ is bounded (Lemma~\ref{lem:sbr:bounded}) and its derivative is $\rateval(\auxpoff-\auxval)$, then $\auxval$ converges to its $\lim\sup$. Moreover, the same reasoning can be made with $\auxpoff[]$ (see its differentiation in (\ref{pr:sbr:diff5})), so it has a limit which is necessarily the same as $\auxval[]$: otherwise, the derivative of $\auxval$ would converge towards $\rateval$ times the difference of the limits and $\auxval$ would be unbounded because of hypothesis~\ref{hyp:ratesbr}.
    \end{proof}

\subsubsection{Convergence of actions}

    \begin{lemma}\label{lem:sbr:cvact}
        Action profile $\auxact$ converges to a fixpoint of $\sbr[\st][\limval][][]$. Therefore, $\auxact[\ti][][]$ converges to a regularized Nash equilibria of the stochastic game.
    \end{lemma}

    \begin{proof}
        We use equation (\ref{pr:sbr:diff5}) to bound above the scalar product:
        \begin{equation}\label{pr:sbr:cvact}
        \langle \sbr[\st][\auxval[\ti][\pl]][\auxact[\plb]][\plb]-\auxact[\ti][\plb], \nabla_\plb \estauxpoff[\ti][\pl][\auxact[\ti]] + \smness \tble\rangle
        \end{equation}
        
        Since $\auxpoff[]$ and $\auxval[]$ have limits and are Lipsichitz (Lemma~\ref{lem:sbr:lip}), then their derivative goes to $0$, and the first term of (\ref{pr:sbr:diff5}) also goes to $0$. As a consequence, the limsup of (\ref{pr:sbr:cvact}) is bounded above by $0$, and this is positive, so it goes to $0$. Therefore, the action profile converges to a fixpoint of $\sbr[\st][\limval][][]$.
    \end{proof}

}

\subsection{Convergence in zero-sum games}\label{app:sbr:zs}
{
    
    In this subsection we suppose that the game is zero-sum, that is there are only two players and for all states, $\poff[\purep][1] = -\poff[\purep][2]$ for all action profiles $\purep$. Therefore, we can omit the superscript with the convention that $\poff[][] = \poff[][1]$.
    
    The proof is inspired from \cite{leslieBestresponseDynamicsZerosum2020}, later extended in \cite{baudinFictitiousPlayBestResponse2022}.

\RenewDocumentCommand{\estauxpoff}{O{\ti} O{\pl} O{\auxact[#1]} O{\st} O{\auxval[#1][][]}}{{\hat f}_{#4, #5}\braced{#3}}

\RenewDocumentCommand{\auxpoff}{O{\ti} O{\pl} O{\auxact[#1]} O{\st} O{\auxval[#1][#2][]}}{f_{#4, #5}\braced{#3}}

\RenewDocumentCommand{\poff}{O{} O{\pl} O{\st}}{r_{#3}\braced{#1}}

\RenewDocumentCommand{\estpoff}{O{\act} O{\pl} O{\st} O{\ti}}{\hat r_{#3}\braced{#1}\braced{#4}}

\RenewDocumentCommand{\auxval}{O{\ti} O{} O{\st}}{u_{#3}\sbraced{#1} }
\RenewDocumentCommand{\dotauxval}{O{} O{} O{\st}}{\dot{u}_{#3} \ifthenelse{\equal{#1}{}}{}{\left( #1 \right)}}

\RenewDocumentCommand{\tble}{O{\auxact} O{}}{h^{#2}\braced{#1}}
    
\NewDocumentCommand{\vu}{O{\st}}{v_{#1, u(t)}}
\newcommand{\fsf}{\auxpoff[\ti][][\auxact[\ti][][s_f(\ti)]][s_f(t)]}
\newcommand{\usf}{\auxval[\ti][][s_f]}
\newcommand{\vsf}{\vu[s_f]}
\newcommand{\tblf}{\smness \tble[\auxact[\ti][][s_f(\ti)]]}

\NewDocumentCommand{\duagap}{O{\ti} O{\st}}{w_{#2}\sbraced{#1}}
\NewDocumentCommand{\boundduagap}{}{D}
    Moreover, we suppose that both players use the same regularizer:
    \begin{equation}
        \tble[x^1_s, x^2_s] = \tble[x^1_s, x^2_s][1] = -\tble[x^1_s,x^2_s][2]
        \label{eq:zsentr}
    \end{equation}
    where both $\tble[][1]$ and $\tble[][2]$ are concave in respectively $x^1_s$ and $x^2_s$.
    
    Therefore, with the same initial conditions, continuation payoffs $\auxval[]$ are opposite for both players, and we also omit the superscript.
    
    \subsubsection{Convergence of payoffs}

        We define the energy of the system, also known as the duality gap (\cite{benaimStochasticApproximationsDifferential2005}):
        \begin{equation}\begin{aligned}
            \duagap = & \max_{\profb[1]\in \simplex{\actions[1]}}\estauxpoff[\ti][1][\profb[1], \auxact[\ti][2]] + \smness \tble[\profb[1], \auxact[\ti][2]] \\
            & - \min_{\profb[2]\in \simplex{\actions[2]}}\estauxpoff[\ti][2][\auxact[\ti][1], \profb[2]] + \smness \tble[\auxact[\ti][1], \profb[2]]
        \end{aligned}
        \end{equation}
        
        This is a positive quantity because the first (second) term is greater (lower) than $\estauxpoff[\ti][1][\auxact[\ti]] + \smness \tble[\auxact[\ti]]$. We are going to show that it is mostly decreasing.

        In the rest of the section, we use a $\ratelim$ such that:
        \begin{equation}
            \ratelim \geq \lim\sup \rateval\label{hyp:ratelim}\tag{H3}
        \end{equation}
        
        A special case is when $\rateval$ goes to $0$, in this case $\ratelim$ can be taken arbitrarily small.

        \begin{lemma}\label{lem:sbr:zs:diffduality} The differential of $\duagap[]$ is bounded:
            $$\diff{\duagap[]} \leq - \ratemin \duagap + \boundduagap \ratelim + \boundduagap \exp\left(-\minsbr \ratemin \ti\right)$$
        \end{lemma}
        \begin{proof}
        \NewDocumentCommand{\maxprof}{O{\ti} O{1}}{y_s^{#2\star}\sbraced{#1}}
            We write $\maxprof[\ti][\pl] = \sbr[\st][\auxval][\auxact][\pl]$ and:
            $$g(\maxprof) :=
            %\max_{\profb[1]\in \simplex{\actions[1]}}\estauxpoff[\ti][1][\profb[1], \auxact[\ti][2]] + \smness \tble[\profb[1], \auxact[\ti][2]] =
            \estauxpoff[\ti][1][\maxprof, \auxact[\ti][2]] + \smness \tble[\maxprof, \auxact[\ti][2]]$$
            where the dependency on $\auxval[], \esttrans[\st][][][], \estpoff[][][\st][], \auxact[]$ is left implicit. This implies:
            $$g(\maxprof) = \max_{\profb[1] \in \simplex{\actions[1]}}
            \estauxpoff[\ti][1][\profb[1], \auxact[\ti][2]] + \smness \tble[\profb[1], \auxact[\ti][2]]$$
            
            Therefore, using the envelope theorem\todo{elaborate}, the derivative of $g$ is written using only derivatives on all variables but $\maxprof$:
            $$\diff {g \circ \maxprof[]} = D_u g \cdot \dot u + D_{x^2_s} g \cdot \dotauxact[][2] + D_{\estpoff[][][\st][]}g \cdot \dotestpoff[][][\st][] + D_{\esttrans[\st][][][]} g \cdot  \dotesttrans$$
            
            With Lemma~\ref{lem:sbr:cvest}, Lemma~\ref{lem:sbr:bounded} and Hypothesis~\ref{hyp:ratelim},
            \begin{equation}
                \diff {g\circ \maxprof[]} \leq \ratelim \| D_u g \| \|\poff\| + D_{x^2_s} g \cdot \dotauxact[][2] + (\|D_{\estpoff[][][\st][]}g\| + \|D_{\esttrans[\st][][][]} g\|) \rateact \minsbr \cvestbound \exp(-\minsbr\ratemin t)
                \label{pr:sbr:zs3}
            \end{equation}
            
            $\estauxpoff[\ti][1][]$ is linear, so:
            \begin{equation}
                \label{pr:sbr:zs2}
            D_{x^2_s} g \cdot \dotauxact[][2] = \estauxpoff[\ti][1][\maxprof, \dotauxact[][2]] + \smness \nabla_{\auxact[][2]} \tble[\maxprof, \auxact[\ti][2]] \cdot \dotauxact[][2]
            \end{equation}
            
            And since $\tble[] = \tble[][1] = -\tble[][2]$ (with (\ref{eq:zsentr})), $\nabla_{\auxact[][2]} \tble[] = -\nabla_{\auxact[][2]} \tble[][2]$ by definition, and $\tble[][2]$ is concave in $x^2_s$, so:
            
            \begin{equation}
            \nabla_{\auxact[][2]} \tble[\maxprof,\auxact[\ti][2]][2] \cdot (\maxprof[\ti][2] - \auxact[\ti][2]) \geq \tble[\maxprof,\maxprof[\ti][2]][2] - \tble[\maxprof,\auxact[\ti][2]][2]
            \label{pr:sbr:zs1}
            \end{equation}
            It follows from (\ref{pr:sbr:zs2}), (\ref{pr:sbr:zs1}), (\ref{eq:zsentr}) and $\dotauxact[][2] = \rateact (\maxprof[\ti][2] - \auxact[\ti][2])$ that
            \begin{equation}\begin{aligned}
                D_{x^2_s} g \cdot \dotauxact[][2] & \leq \rateact \estauxpoff[\ti][1][\maxprof, \maxprof[\ti][2] - \auxact[\ti][2]] + \smness \rateact \tble[\maxprof,\maxprof[\ti][2]] - \rateact \smness \tble[\maxprof,\auxact[\ti][2]] \\
                & \leq - \rateact g(\maxprof) + \rateact \estauxpoff[\ti][1][\maxprof, \maxprof[\ti][2]] + \smness \rateact \tble[\maxprof,\maxprof[\ti][2]]
                \end{aligned}
                \label{pr:sbr:zs5}
            \end{equation}
            
            Since $g$ as a function of $\auxval[], \esttrans[\st][][][], \estpoff[][][\st][], \auxact[]$ is Lipschitz, and using (\ref{pr:sbr:zs3}) and (\ref{pr:sbr:zs5}), there exists $\boundduagap$ such that:
            \begin{multline}
            \diff {g \circ \maxprof[]} \leq \frac{\boundduagap}{2} \ratelim + \frac{\boundduagap}{2} \rateact \exp(-\minsbr \ratemin \ti) - \rateact g(\maxprof) \\ + \rateact \estauxpoff[\ti][1][\maxprof, \maxprof[\ti][2]] + \smness \rateact \tble[\maxprof,\maxprof[\ti][2]]
            \label{pr:sbr:zs6}
            \end{multline}

            The same reasoning apply for the second term of $\duagap[]$ with the opposite payoff function (notice that in this case, the second line of (\ref{pr:sbr:zs6}) is exactly the opposite, therefore when summed it cancels out), leading to, when summed:
            \begin{equation*}
                \begin{aligned}
                \diff {\duagap[]} & \leq \boundduagap \ratelim + \boundduagap \rateact \exp(-\minsbr \ratemin \ti) - \rateact \duagap \\
                & \leq \boundduagap \ratelim + \boundduagap \exp(-\minsbr \ratemin \ti) - \ratemin \duagap
                \end{aligned}
            \end{equation*}
            because $\ratemin \leq \rateact \leq 1$ and $\duagap[] \geq 0$ (its first term is greater than $\estauxpoff[\ti][1][\auxact[\ti]] + \smness \tble[\auxact[\ti]]$ and its second one is lower).
        \end{proof}
        
        The following lemma implies that the auxiliary payoff is close to the value of the auxiliary game because the duality gap of the auxiliary game is small enough.
        
        \begin{lemma}\label{lem:sbr:zs:cvduality}
            For all states $\st \in \states$, $$\lim\sup \duagap \leq 2 \boundduagap \ratelim \ratemin^{-1}$$
            Furthermore, $\max\{\duagap[]-2\boundduagap, 0\}$ is a Lyapunov function of \ref{eq:sbrd} (i.e., when payoff and transitions estimate are exact) in the autonomous case.
        \end{lemma}
        
        \begin{proof}
            Since $\duagap[]$ is positive (the first term is greater than $\estauxpoff[\ti][1][\auxact[\ti]] + \smness \tble[\auxact[\ti]]$ and the second one is lower), Lemma~\ref{lem:sbr:zs:diffduality} makes it possible to use Grönwall's Lemma on $\duagap[] - 2\boundduagap \ratelim \ratemin^{-1}$.  as soon as $\exp(-\minsbr \ratemin \ti) \leq \ratelim$.
            
        \end{proof}

        \NewDocumentCommand{\eps}{}{\xi}
        
        Define $\eps$ such as $\frac{(1-\delta)\eps}{16} = 4 \boundduagap \ratelim \ratemin^{-1}$.

        Estimates of transitions and payoffs are close to real values for $t$ large enough, so Lemma~\ref{lem:sbr:zs:cvduality} implies that there exists $t_1(\eps)$ such that for $\ti \geq t_1(\eps)$:
        \begin{equation}
            \begin{aligned}
                |\estauxpoff[][][]-\auxpoff[][][]| \leq 4 \boundduagap \ratelim \ratemin^{-1} = \frac{(1-\delta)\eps}{16}\\
                |\estauxpoff[][][]+\smness \tble[] - \vu| \leq 2 \boundduagap \ratelim \ratemin^{-1} = \frac{(1-\delta)\eps}{32} \\
                |\auxpoff[][][]+\smness \tble[]-\vu| \leq 2 \boundduagap \ratelim \ratemin^{-1} = \frac{(1-\delta)\eps}{32}
            \end{aligned}
            \label{eq:sbr:hypcv}
            \tag{A1}
        \end{equation}
        where $\vu$ is the value of the auxiliary game parameterized by $u(t)$ (and functions $\tble[]$, $\auxpoff[]$ are $\estauxpoff[]$ are valued at $\auxact$, omitted for readability).

        We define two distinguished states (notice that we use $\auxpoff[][][]$ and not $\estauxpoff[][][]$):
        \begin{itemize}
            \item $s_f(t) \in \argmax{s\in \states}|\auxpoff + \smness \tble - \auxval|$
            \item $s_v(t) \in \argmax{s\in \states}|\vu - \auxval|$
        \end{itemize}

        \begin{lemma}\label{lem:usf}
            If (\ref{eq:sbr:hypcv}) is satisfied (for instance if $t \geq t_1(\eps)$) and
            $$|\usf - \fsf - \tblf| \geq \eps$$
            and for an $s \in \states$, $$\left||\usf - \vsf| - |\auxval[s] - \vu|\right| \leq \frac{(1-\delta)\eps}{8}$$
            then:
            $$\diff{|\auxval[s] - \vu|}\leq - \frac{(1-\delta)\rateval \eps}{2}$$
        \end{lemma}

        \begin{proof}

            First, using Lemma~A.2 of \cite{leslieBestresponseDynamicsZerosum2020} on the regularity of the value of a zero-sum (static) game, it follows:
            \begin{equation}
                \begin{aligned}\left|\diff {\vu}\right| & \leq \delta \max_{s\in \states} | \dotauxval| \\
                &= \delta \rateval |\fsf + \tblf - \usf| + \delta \rateval  \frac{(1-\delta)\eps}{16} \\
                & \leq \delta \rateval \eps(1+\frac{1-\delta}{16})
                \end{aligned}
                \label{pr:sbr:zs7}
            \end{equation}
            
            We now prove that $\auxval$ moves towards $\vu$ at a constant speed relatively to $\rateval$:
    
            \begin{itemize}
    
            \item If $u_s(t) \geq \vu$, then $|\usf-\vsf|-\auxval+\vu \leq \frac{(1-\delta)\eps}{8}$.

            $$\begin{aligned}
            \dotauxval & = \rateval \left(\estauxpoff + \smness \tble - \auxval\right) \\
            &\leq \rateval \left(\auxpoff + \smness \tble - \auxval \right) +\rateval \frac{(1-\delta)\eps}{16} \\
            & \leq \rateval \left(\auxpoff + \smness \tble + \frac{(1-\delta)\eps}{8}-\vu-|\usf-\vsf|\right) \\ & \ \ \ \ + \rateval \frac{(1-\delta)\eps}{16}\\
            & \leq \rateval \left(\frac{4(1-\delta)\eps}{16}-|\usf-\vsf|\right) \\
            & \leq \rateval \left(\frac{(1-\delta)\eps}{4}-|\usf-\fsf-\tblf|\right) \\
            & \leq \rateval \left(\frac{(1-\delta)\eps}{4}-\eps\right)
            \end{aligned}$$
    
            Summing with $\vu$ and using (\ref{pr:sbr:zs7}):
            $$
            \begin{aligned}
            \diff{\auxval-\vu} & \leq\rateval \left(\frac{(1-\delta)\eps}{4} - \eps + \delta \eps + \delta \eps \frac{1-\delta}{16}\right) \\
            & \leq \rateval \eps \left( \frac{1-\delta}{2} -1 + \delta\right) \\
            & \leq - \rateval \left(\frac{2(1-\delta)\eps}{4}\right) \\
            \end{aligned}
            $$

            \item If $u_s(t) \leq \vu$, similar calculations yield the same result.
            \end{itemize}

        \end{proof}

        \begin{lemma}
            For all $s\in \states$, $\lim\sup_{t\rightarrow \infty} |\auxval-\auxpoff-\smness \tble| \leq 4\eps$.
        \label{lem:sbr:zs:cvval}
        \end{lemma}
        \begin{proof}
            We define $g(t) = \max\{|\usf-\vsf|, 3\eps\}$.
    
            %As a composition of maximum of locally Lipschitz function and as such is locally Lipschitz as well. (See Lemma~B.4 of \cite{leslieBestresponseDynamicsZerosum2020} for detailed arguments.)
    
            Now, if $|\usf-\vsf|\leq 2 \eps$, then $\diff g = 0$. If $|\usf-\vsf| \geq 2 \eps$ and if $t$ is greater than $t^1(\eps)$, then $|\usf-\fsf-\tblf| \geq \eps$: indeed, $|\fsf + \tblf - \vsf| \leq \eps$ because of Lemma~\ref{lem:sbr:zs:cvduality} and its corollary \ref{eq:sbr:hypcv}. Similarly, on a neighbourhood of $t$, every $s$ that maximizes $|\auxpoff + \smness \tble  - \auxval|$ satisfies the condition of Lemma~\ref{lem:usf}, because $|\auxpoff+\smness \tble -\vu| \leq \eps$ according to the same Lemma~\ref{lem:sbr:zs:cvduality}. Therefore, Lemma~\ref{lem:usf} can be used and: $$\diff g \leq- \frac{3(1-\delta)\rateval \eps}{4}$$
            
            This holds as soon as $g(t) > 2\eps$ and $t > t^1(\eps)$. The integral of $\rateval[]$ is infinite (hypothesis \ref{hyp:ratesbr}), so there is a $t^2(\eps)$ such that for $t \geq t^2(\eps)$, $g(t) = 2\eps$.
            
            Then, using \ref{eq:sbr:hypcv}, we have $|\usf-\fsf-\tblf| \leq 3\eps$ and by definition of $s_f$, the inequality of the lemma.

        \end{proof}

    \subsubsection{Convergence of actions}
        \begin{lemma}
        For all $s \in \states$, $x_s(t)$ converge to the set of $3\eps$-Regularized Nash equilibria of the auxiliary game.
        \end{lemma}

        \begin{proof}
        The previous proof gives that $\auxpoff$ is $3\eps$ close to $\vu$, hence the result.
        \end{proof}
}

\section{Convergence of smooth fictitious-play in identical-interest and zero-sum stochastic games}\label{app:sfp}

\NewDocumentCommand{\ict}{}{L}
\NewDocumentCommand{\firstict}{}{A}
\NewDocumentCommand{\secondict}{}{B}

In order to prove the convergence of the discrete-time procedures described in this paper, we use the theory of stochastic approximations to relate continuous-time and discrete-time systems. The proof is in two steps: first we characterize the internally chain transitive sets of the continuous-time systems (definition below), then we show with stochastic approximations theorem that limit sets of the discrete time systems are included in internally chain transitive sets.

\begin{definition}[Internally chain transitive]
A set $A$ is internally chain transitive for a differential inclusion $\diff y \in F(y)$ if it is compact and if for all $x, x' \in A$, $\epsilon > 0$ and $T > 0$ there exists an integer $n\in \mathbb N$, solutions $y_1, \ldots y_n$ to the differential inclusion and real numbers $t_1, t_2, \ldots, t_n$ greater than $T$ such that:
\begin{itemize}
    \item $y_i(s) \in A$ for $0\leq s \leq t_i$
    \item $\|y_i(t_i)-y_{i+1}(0)\|\leq \epsilon$
    \item $\|y_1(0)-x| \leq \epsilon$ and $\|y_n(t_n)-x'\|\leq \epsilon$
\end{itemize}
\label{def:ict}
\end{definition}

The previous definition means that if a set $A$ is internally chain transitive, two points can be linked with solutions of at least length $T$ in at most $n$ steps, where $n$ depends on $T$. These sets contain the limit sets of the discrete-time counterparts of the differential inclusions (see (\cite{benaimStochasticApproximationsDifferential2005}) for an introduction to the theory of stochastic approximations and formal statements), that is systems of the form:
$$y_{n+1}-y_n \in \beta_n (F(y_n)+ U_n)$$
where $U_n$ is typically a zero-mean noise with bounded variance.

In the following, we characterize internally chain transitive sets and use an asynchronous extension of the theory of stochastic approximations initially presented in \cite{perkinsAdvancedStochasticApproximation2013} and later extended in \cite{baudinFictitiousPlayBestResponse2022}.

\

\subsection{Payoff perturbation}

    In \refmbrd, we supposed that players did not observe the actual stage reward $r^i_{s_n}(a_n)$ but a perturbation of this value. Formally, the expectancy of $R^i_n$ must be $r^i_{s_n}(a_n)$ conditionally on the history:
    \begin{equation}
        \mathbb E \left[ R^i_n | \mathcal F_{n-1}\right] = r^i_{s_n}(a_n)
    \end{equation}
    where $\mathcal F_{n-1}$ is the $\sigma$-algebra that contains all information up to step $n-1$. Furthermore, we require the variance to be bounded, i.e. $\operatorname{var}(R^i_n | \mathcal F_{n-1})$ is bounded.

\subsection{Convergence of estimates}{
    
    \begin{lemma}\label{lem:sfp:firstict}
        If $\ict$ is an internally chain transitive set of \refmbrd, then it is included in $\firstict$, where:
    $$\firstict := \left\{(\auxact[][\pl], \auxval[][\pl], \estpoff[][\pl][\st][], \esttrans[\st][][][])_{\st, \pl}\ \left|\ \forall \st, \estpoff[][\pl][\st][] = \poff[]
    \land \esttrans[\st][][][]=\trans[\st][][] \right.\right\}$$
    \end{lemma}
    \begin{proof}
        We notice that functions $\estpoff[][\pl][\st][] \mapsto \|\estpoff[]-\poff[]\|_{\infty}$ and $\esttrans[\st][][][] \mapsto \|\esttrans-\trans[\st][][]\|_{\infty}$ are Lyapunov functions with calculations similar to Lemma~\ref{lem:sbr:cvest}: their derivative is smaller than $-\ratemin$ multiplied by their value. This implies that $\ict$ is contained in the inverse of zero of such functions, that is $\firstict$.
    \end{proof}

} % end convergence estimates

\subsection{Convergence in identical-interest stochastic games}{

    \begin{lemma}\label{lem:sfp:ii:ict}
        Suppose that the system is autonomous, that is $\rateval$ is constant and we suppose $\rateval = 1$ (it can be generalized to any constant). Let $\ict$ be an internally chain transitive set of \refmbrd, then:
        \begin{equation}\label{lem:sfp:myict}
            \ict \subseteq \left\{(\auxact[], \auxval[], \estpoff[][][\st][], \esttrans[\st][][][])_{\st}\ \left|\ \forall \st, \auxpoff[] +\smness\tble[\auxact[]]= \auxval[]
            \land \auxact[] = \sbr[\st][\auxval[]][\auxact[]]
            \land \estpoff[][][\st][] = \poff[]
            \land \esttrans[\st][][][]=\trans[\st][][]
            \right. \right\}
        \end{equation}
        which means that $\ict$ contains only regularized equilibria.
    \end{lemma}
    
    \begin{proof}{

        The proof proceeds with a sequence of inclusion. We show that any internally chain transitive set is contained in:
        $$\firstict := \left\{(\auxact[], \auxval[], \estpoff[][][\st][], \esttrans[\st][][][])_{\st}\ \left|\ \forall \st, \estpoff[][][\st][] = \poff[]
        \land \esttrans[\st][][][]=\trans[\st][][] \right.\right\}$$
        $$\secondict := \left\{(\auxact[], \auxval[], \estpoff[][][\st][], \esttrans[\st][][][])_{\st}\ \left|\ \forall \st, \auxpoff[]+\smness\tble[\auxact[]] \geq \auxval[] \right.\right\}$$
        and then in the set of Eq. (\ref{lem:sfp:myict}).

        \NewDocumentCommand{\abbrauxpoff}{O{\ti} O{} O{} O{\st} O{}}{\Gamma_{#4}\braced{#1}}
        \NewDocumentCommand{\stm}{O{\ti}}{s_-\braced{#1}}
        \NewDocumentCommand{\argmin}{m}{\arg\min_{#1}}

        {
        \RenewDocumentCommand{\auxpoff}{}{\abbrauxpoff}
        \NewDocumentCommand{\mauxval}{O{\ti} O{#1} O{\stm[#2]}}{\auxval[#1][][#3]}
        \NewDocumentCommand{\mauxpoff}{O{\ti} O{#1} O{\stm[#2]}}{\auxpoff[#1][][][#3]}
    
        Regarding $\ict \subseteq \firstict$, this is exactly Lemma~\ref{lem:sfp:firstict} when payoffs functions are equal.
    
        In (\ref{pr:sbr:diff3}), the term $-2\cvestbound\exp(-\minsbr \ratemin \ti)$ comes from the fact that $\estpoff[]$ and $\esttrans$ have not converged yet. Therefore, if we are in $\firstict$, these terms disappears and the new differential inequality resulting from computations of Lemma~\ref{lem:sbr:inegdiff} is:
        \begin{equation*}
            \diff{\mauxpoff[]-\mauxval[]} \geq \rateval (\disc-1)(\mauxpoff-\mauxval)
        \end{equation*}
        which, using Grönwall lemma and $\rateval=1$, implies that:
        \begin{equation*}
            \mauxpoff-\mauxval \geq \left(\mauxpoff[0]-\mauxval[0]\right)\exp\left( -(1-\delta)\ti  \right)
        \end{equation*}
    
        Therefore, for all states $\st$:
        \begin{equation}\label{pr:sbr:gap0}
            \auxpoff-\auxval \geq \left(\mauxpoff[0]-\mauxval[0]\right)\exp\left( -(1-\disc)\ti   \right)
        \end{equation}
        
        \NewDocumentCommand{\dist}{}{\epsilon}
        \NewDocumentCommand{\lensol}{}{T}
    
        Then, let $a$ be a point of $\ict$ with for all $\st$, $\auxpoff[]-\auxval[] \geq -\dist$ with equality for a state (and $\dist > 0$). We suppose that the lipschitz constant of $\auxpoff[][][]$ is $1$, other cases are analoguous. By definition, it is linked to itself by solutions of \refmbrd which are collated with an arbitrary gap, we can take $\dist/2$ and of length at least an arbitrary $\lensol$, take $\log(4)/(1-\disc)$. Now, after the first solution, by (\ref{pr:sbr:gap0}), the difference between the auxiliary payoff $\auxpoff[]$ and the auxiliary value $\auxval[]$ is greater than $-3\dist/2\exp(-\log(4))$ (because the initial value of $\auxpoff[]-\auxval[]$ is greater than $-\dist-\dist/2$ and then using (\ref{pr:sbr:gap0})). This implies, by recurrence, that the difference between $\auxpoff[]-\auxval[]$ is strictly greater than $-\dist$, which is absurd since the end of the chain is $a$.\todo{à développer}

        Therefore, $\ict \subseteq \secondict$.
    
        Then, relatively to $\firstict \cap \secondict$, $\auxpoff[]$ is a Lyapunov function because $\diff{\auxpoff[]}$ is greater than $0$ based on (\ref{pr:sbr:diff5}), and strictly greater than $0$ for points outside of the set of (\ref{lem:sfp:myict}), which concludes the lemma.\todo{pourrait être développé}
        
        }

    }\end{proof}
    
    \begin{proof}[Proof of Theorem~\ref{thm:sfp:team}]{
        
\RenewDocumentCommand{\estauxpoff}{O{\ti} O{\pl} O{\auxact[#1]} O{\st} O{\auxval[#1][][]}}{{\hat f}_{#4, #5}\braced{#3}}

\RenewDocumentCommand{\auxpoff}{O{\ti} O{\pl} O{\auxact[#1]} O{\st} O{\auxval[#1][#2][]}}{f_{#4, #5}\braced{#3}}

\RenewDocumentCommand{\poff}{O{} O{\pl} O{\st}}{r_{#3}\braced{#1}}

\RenewDocumentCommand{\estpoff}{O{\act} O{\pl} O{\st} O{\ti}}{\hat r_{#3}\braced{#1}\braced{#4}}

\RenewDocumentCommand{\auxval}{O{\ti} O{} O{\st}}{u_{#3}\sbraced{#1} }
\RenewDocumentCommand{\dotauxval}{O{} O{} O{\st}}{\dot{u}_{#3} \ifthenelse{\equal{#1}{}}{}{\left( #1 \right)}}

\RenewDocumentCommand{\tble}{O{\auxact} O{}}{h^{#2}\braced{#1}}
        \RenewDocumentCommand{\rateval}{}{\beta}
        With Lemma~\ref{lem:sfp:ii:ict}, we known that ICT sets of \refmbrd are contained in the set of regularized equilibria. The rest of the proof uses stochastic approximations results to show that ICT sets contains the limit sets of \ref{eq:mfp} and \ref{eq:sfp}.
    
        In order to do this, we use Theorem~D.5\todo{à préciser} of \cite{baudinFictitiousPlayBestResponse2022}, an extension of Theorems of \cite{perkinsAsynchronousStochasticApproximation2012} and \cite{benaimStochasticApproximationsDifferential2005}.
    
        Let us recall the two systems that we want to relate:
    
        \begin{equation}\label{eq:mfp2}\tag{MFP}
            \left\{
                \begin{aligned}
                    & \ddotauxval = \drateval \left(\destauxpoff + \smness \dtble - \dauxval\right) \\
                    & \ddotauxact = \drateact \left(\dact - \dauxact\right) \\
                    & \desttrans[\st][s'][a][\nti] - \desttrans[\st][s'][a] = \frac{1_{\curst[\dti]=\st\land \act[\dti]=a}}{\sum_{k=0}^{\dti} 1_{\curst[k]=\st\land \act[k]=a}}\left(1_{\curst[\nti]=s'}  - \desttrans[\st][s'][a]\right) \\
                    %\frac{\sum_{k=0}^{\dti} 1_{\curst[k]=\st\land \act[k]=a} 1_{\curst[k+1]=s'} }{\sum_{k=0}^{\dti} 1_{\curst[k]=\st\land \act[k]=a}}\\
                    & \destpoff[a][][\st][\nti] - \destpoff[a][][\st][\dti]= \frac{1_{\curst[\dti]=\st\land \act[\dti]=a}}{\sum_{k=0}^{\dti} 1_{\curst[k]=\st\land \act[k]=a}}\left(\tipoff[\dti]  -\destpoff[a][][\st][\dti]\right)\\
                    & \dact[\dti][\pl] \sim \dsbr
                \end{aligned}
            \right.
        \end{equation}
        which is MFP rewritten with incremental updates for estimators, and
        
        \begin{equation}\left\{
            \begin{aligned}
                & \dotauxval  = \rateval \left(\estauxpoff + \smness \tble -\auxval\right) \\
                & \dotauxact  = \rateact \left(\act - \auxact\right) \\
                & \dotesttrans[\st][s'][b] = \rateact \act (b) \left( \trans[\st][s'][b] - \esttrans[\st][s'][b]\right)\\
                & \dotestpoff[b] = \rateact \act (b) \left( \poff[b]  - \estpoff[b]\right)\\
                & \act[\ti][\pl] = \sbr \\
                & \rateact  \in [\ratemin, 1] \\ 
            \end{aligned}
        \right.\tag{MBRD}\label{eq:mbrd2}\end{equation}
        
        The first system is of the form:
        \begin{equation}
            y_{\nti} - y_\dti \in S_n \cdot (F(y_\dti) + U_\dti)
        \end{equation}
        where $\cdot$ is the pairwise multiplication of two vectors, and the second one is:
        \begin{equation}
            \diff y \in S_\ti \cdot F(y)
        \end{equation}
        
        Note that this is the same $F$ between two systems and that the expectancy of $U_\dti$ is zero. Vectors $S_\dti$ and $S_\ti$ are the update rates of every variable. For every discrete variable that are updated with an indicator function, the steps are decreasing in $\frac{1}{\dti}$ in the number of times that the indicator was equal to $1$. For $\dauxval$, it is also updated in $\frac{1}{\dti}$. Therefore, our update step sequence is $\gamma_{\dti} = \frac{1}{\dti}$. Furthermore, update steps $S_\dti$ and $S_\ti$ are correlated vectors, meaning that the update rate of every variable is not independent (for instance, for a fixed state and a fixed action, $\destpoff[a][][\st][\ti]$ and $\desttrans[\st][s'][a][\ti]$ are updated at the same rate. So this is a correlated asynchronous system, as described in \cite{baudinFictitiousPlayBestResponse2022}.
        
        Now, we must verify every hypothesis of the theorem:
            \begin{enumerate}[label=(\roman*)]
                \item All values of \ref{eq:mfp} are bounded, so they belong to a compact.
                \item The right hand side F of \refmbrd is a Marchaud map because it has bounded, closed and convex images, and its graph is itself closed.
                \item We use steps $\gamma_n = \frac{1}{n}$ which satisfy the usual properties: it is decreasing, it goes to $0$ and the sum is equal to $\infty$.
                \item The game is ergodic, and the next states are randomly chosen based uniquely on the current state and history, so properties on the transitions are verified.
                \item There is a correlation between $U_n$ and $S_n$, therefore it is necessary to prove the Kushner-Clark noise condition separately (see \cite{perkinsAsynchronousStochasticApproximation2012}), one variable at a time. Regarding $\desttrans$ and $\destpoff[]$, this is true because the noise is uniformly bounded in $L^2$ (bounded variance) and the update step is $1/t$, as noted in Proposition~1.4 and Remark~1.5 of \cite{benaimStochasticApproximationsDifferential2005}. Regarding variables $\dauxact$ and $\dauxval$ there is no correlation between the update steps and the noise, so standard theorems apply (for instance Lemma 3.3 of \cite{perkinsAsynchronousStochasticApproximation2012}).
                \item The squared sum of $\gamma_n$ converges and the noise is bounded.
                \item There is no additional drift, so $d_n = 0$ in our case.
            \end{enumerate}
        
            Therefore, we can use the theorem, which implies that the limit set of \ref{eq:mfp} is an internally chain transitive set of \refmbrd. Combined with Lemma~\ref{lem:sfp:ii:ict}, this gives the desired result.
    
    }\end{proof}

} % end team games

\subsection{Convergence in zero-sum stochastic games} {

    The scheme of the proof is similar to that of identical-interest stochastic games: first, we characterize internally chain transitive sets of \refmbrd, then we show that systems \refmbrd and \ref{eq:mfp} can be related using stochastic approximations theory, thus the characterization of limit sets of \ref{eq:mfp}. Throughout this part, we use notations and hypothesis from subsection~\ref{app:sbr:zs}.
    
    \NewDocumentCommand{\boundict}{}{E}
    
    \begin{lemma}[Internally chain transitive sets in ZS case] Let $\ict$ be an internally chain transitive sets of system \refmbrd in the ZS case with the same initial values for all players. Then there exists $\boundict$ such that $\ict$ is contained in the following set:
    
        $$\secondict^{\boundict\ratelim} := \left\{(\auxact[], \auxval[], \estpoff[][][\st][], \esttrans[\st][][][])_{\st}\ \left|\ \forall \st, |\auxpoff[]+\smness\tble[\auxact[]]-\auxval[]| \leq \boundict \ratelim \right.\right\}$$
        \label{lem:sfp:zs:ict}
    \end{lemma}
    
    \begin{proof}
        With notations of subsection~\ref{app:sbr:zs}, we use the same function $g(\auxval[]) = \max\{|\usf-\vsf|, 2\epsilon\}$ as in Lemma~\ref{lem:sbr:zs:cvval}. Furthermore, assumption (\ref{eq:sbr:hypcv}) is satisfied for points of $\ict$ because $\ict$ is contained in $\firstict$ (Lemma~\ref{lem:sfp:firstict}) and the duality gap $\duagap$ is guaranteed to be small enough. Indeed, $\duagap$ defined in subsection \ref{app:sbr:zs} is almost a Lyapunov function: relatively to set $\firstict$, $\duagap[]-2\boundduagap\ratelim\minsbr^{-1}$ is a Lyapunov function (Lemma~\ref{lem:sbr:zs:cvduality}). Therefore, Lemma~\ref{lem:usf} implies that $g$ is Lyapunov as well (Lemma~\ref{lem:sbr:zs:cvval}) and this implies that $\ict \subseteq \secondict^{\boundict\ratelim}$.\todo{elaborate}
    \end{proof}
    
    \begin{proof}[Proof of Theorem~\ref{thm:sfp:zerosum}]
        Theorem of stochastic approximations apply identically to the identical-interest stochastic games case in the previous subsection. Therefore, with the characterization of internally chain transitive sets (Lemma~\ref{lem:sfp:zs:ict}), we have an inclusion for the limit sets of \ref{eq:mfp} in the zero-sum case as well.
    \end{proof}

} % end zero-sum games

\iffalse

\section{Stochastic Approximations}
\label{app:stochasticapprox}
\input{appendix/c_stochastic_approximations.tex}

\fi

% Next step : problème avec h, il faut que sbr soit éloigné de 0 pour qu'on explore suffisamment.
% on refait toute les preuves sur papier dans l'ordre, on réécrit le tout en latex ensuite

\end{document}